%% file: arxiv-main.tex
\documentclass[11pt]{article}
\input{preamble}
\begin{document}

\title{Streaming Algorithms with Few State Changes}
\date{\today}
\author{
Rajesh Jayaram\thanks{Google Research. 
E-mail: \email{rkjayaram@google.com}}
\and
David P. Woodruff\thanks{Carnegie Mellon University and Google Research. 
E-mail: \email{dwoodruf@cs.cmu.edu}}
\and
Samson Zhou\thanks{Texas A\&M University. Work done in part while at UC Berkeley and Rice University.  
E-mail: \email{samsonzhou@gmail.com}}
}

\maketitle

\begin{abstract}
In this paper, we study streaming algorithms that minimize the number of changes made to their internal state (i.e., memory contents). While the design of streaming algorithms typically focuses on minimizing space and update time, these metrics fail to capture the asymmetric costs, inherent in modern hardware and database systems, of reading versus writing to memory. In fact, most streaming algorithms write to their memory on \emph{every update}, which is undesirable when writing is significantly more expensive than reading. This raises the question of whether streaming algorithms with small space \emph{and} number of memory writes are possible.

We first demonstrate that, for the fundamental $F_p$ moment estimation problem with $p\ge 1$, any streaming algorithm that achieves a constant factor approximation must make $\Omega(n^{1-1/p})$ internal state changes, regardless of how much space it uses. Perhaps surprisingly, we show that this lower bound can be matched by an algorithm that also has near-optimal space complexity. Specifically, we give a $(1+\varepsilon)$-approximation algorithm for $F_p$ moment estimation that uses a near-optimal $\widetilde{\mathcal{O}}_\varepsilon(n^{1-1/p})$ number of state changes, while simultaneously achieving near-optimal space, i.e., for $p\in[1,2]$, our algorithm uses $\text{poly}\left(\log n,\frac{1}{\varepsilon}\right)$ bits of space, while for $p>2$, the algorithm uses $\widetilde{\mathcal{O}}_\varepsilon(n^{1-2/p})$ space. We similarly design streaming algorithms that are simultaneously near-optimal in both space complexity and the number of state changes for the heavy-hitters problem, sparse support recovery, and entropy estimation. Our results demonstrate that an optimal number of state changes can be achieved without sacrificing space complexity. 
\end{abstract}

\section{Introduction}
The streaming model of computation is a central paradigm for computing statistics for datasets that are too large to store. Examples of such datasets include internet traffic logs, IoT sensor networks, financial transaction data, database logs, and scientific data streams (such as huge experiments in particle physics, genomics, and astronomy). 
In the one-pass streaming model, updates to an underlying dataset are processed by an algorithm one at a time, and the goal is to approximate, collect, or compute some statistic of the dataset while using space that is sublinear in the size of the dataset (see \cite{babcock2002models,muthukrishnan2005data} for surveys).

Formally, an insertion-only data stream is modeled by a sequence of updates $u_1,\ldots,u_m$, each of the form $u_t\in[n]$ for $t\in[m]$, where $[n]=\{1,\ldots,n\}$ is the universe size. 
The updates implicitly define an underlying frequency vector $f\in\mathbb{R}^n$ by $f_i=|\{t\,\mid\,u_t=i\}|$, so that the value of each coordinate of the frequency vector is the number of occurrences of the coordinate identity in the data stream. 

One of the most fundamental problems in the streaming literature is to compute a $(1+\eps)$ approximation of the $F_p$ moment of the underlying frequency vector $f$, defined by $F_p(f)=(f_1)^p+\ldots+(f_n)^p$, where $\eps>0$ is an accuracy parameter. 
The frequency moment estimation problem has been the focus of more than two decades of study in the streaming model~\cite{AlonMS99,Bar-YossefJKS04,Woodruff04,IndykW05,Indyk06,Li08,KaneNW10,LiW13,BlasiokDN17,BravermanVWY18,GangulyW18,JayaramW18b,JayaramW19,WoodruffZ21,WoodruffZ21b,AjtaiBJSSWZ22,Ben-EliezerJWY22}. 
In particular, $F_p$-estimation is used for $p=0.25$ and $p=0.5$ in mining tabular data~\cite{CormodeIKM02}, for $p=1$ in network traffic monitoring~\cite{FeigenbaumKSV02} and dynamic earth-mover distance approximation~\cite{Indyk04}, and for $p=2$ in estimating join and self-join sizes~\cite{AlonGMS02} and in detecting network anomalies~\cite{ThorupZ04}. 
Note that a $(1+\eps)$-approximation to the $F_p$ moment also gives a $(1+\cO{\eps})$-approximation to the $L_p$ norm of $f$, defined by 
\[L_p(f)=\|f\|_p=\left(F_p(f)\right)^{1/p}=\left((f_1)^p+\ldots+(f_n)^p\right)^{1/p}.\]

Another fundamental streaming problem is to compute $L_p$-heavy hitters: given a threshold parameter $\eps\in(0,1]$, the $L_p$-heavy hitters problem is to output a list $L$ containing all $j\in[n]$ such that $f_j\ge\eps\cdot\|f\|_p$, and no $j\in[n]$ with $f_j<\frac{\eps}{2}\cdot\|f\|_p$. 
The heavy-hitter problem is used for answering iceberg queries~\cite{FangSGMU98} in database systems, finding elephant flows and spam prevention in network traffic monitoring~\cite{Ben-BasatEFK17a}, and perhaps has an even more extensive history than the $F_p$ moment estimation problem~\cite{MisraG82,Moore91,CharikarCF04,CormodeM05,MetwallyAA05,MankuM12,Indyk13,BravermanCIW16,LarsenNNT16,BravermanCINWW17,BravermanGLWZ18,LiNW18,IndykNW22,BlockiLMZ23,LebedaT23}.

The primary goal of algorithmic design for the streaming model is to minimize the space and update time of the algorithm. 
However, the generic per-update processing time fails to capture the nuanced reality of many modern database and hardware systems, where the \emph{type} of updates that are made on a time step matter significantly for the real-world performance of the algorithm. 
Specifically, it is typically the case that updates that only require \textit{reads} to the memory contents of the algorithm are significantly faster than updates that also modify the memory of the algorithm, i.e., writes. 
Thus, while many streaming problems are well-understood in terms of their space and update time, little is known about their \textit{write} complexity: namely, the number of state changes made over the course of the stream. 

In this paper, we propose the number of state changes of a streaming algorithm as a complexity-theoretic parameter of interest, and make the case for its importance as a central object of study, in addition to the space and update-time of an algorithm. 
While there is significant practical motivation for algorithms that update their internal state infrequently (see \secref{sec:motivation} for a discussion), from a theoretical perspective it is not clear that having few state changes is even possible. 
Specifically, most known streaming algorithms write to their memory contents on \emph{every} update of the stream. 
Moreover, even if algorithms using fewer state changes existed, such algorithms would not be useful if they required significantly more space than prior algorithms that do not minimize the number of state changes. 
Therefore, we ask the following question:

    \begin{quote}
    \begin{center}
       {\it Is it possible to design streaming algorithms that make few updates to their internal state in addition to being space-efficient? }
        \end{center}
    \end{quote}

Our main contribution is to answer this question positively. 
Specifically, we demonstrate that algorithms exist that are simultaneously near-optimal in their space and state-change complexity. 
This demonstrates further that we do not need to pay extra in the space complexity to achieve algorithms with few state changes. 


\subsection{Motivation for Minimizing State Changes}
\seclab{sec:motivation}

\paragraph{Asymmetric read/write costs of non-volatile memory.}
The primary motivation of minimizing state changes arises from the simple observation that different actions performed over the same allocated memory may have different costs. 
For example, non-volatile memory (NVM) is low latency memory that can be retained when power to the system is turned off and therefore can dramatically increase system performance. 
Although NVM offers many benefits over dynamic random access memory (DRAM), writing to NVM is significantly more costly than reading from NVM, incurring higher energy costs, higher latency, and suffering lower per-chip bandwidth~\cite{BlellochFGGS15}. 
In fact, \cite{MeenaSCT14} noted that an NVM cell generally wears out between $10^8$ and $10^{12}$ writes. 
Hence, in contrast to DRAM, NVM has a significant asymmetry between read and write operation costs~\cite{DongWSXLC08,DongJX09,AkelCMGS11,2011Qureshi,XuDJX11,AthanassoulisBCR12,KimSDC14}, which has been the focus of several works in system design~\cite{YangLKCLY07,ChoL09,LeeIMB09,ZhouZYZ09} and database data structure design~\cite{ChenGN11,Viglas12,Viglas14}. 

One specific type of non-volatile memory is NAND flash memory, which is an electronic memory storage unit that can be electrically erased, written, and read in blocks that are generally significantly smaller than the entire device. 
NAND flash memory can be found in a number of common devices, such as smartphones, USB flash drives, memory cards, or solid-state drives. 
However, \cite{Ben-AroyaT11} notes that, at the time, individual NAND flash memory cells would wear out after $10^4$ to $10^6$ write/erase operations. 
Indeed, Apple notes that ``all iOS devices and some macOS devices use a solid-state drive (SSD) for permanent storage'' and recommends minimizing disk writes to optimize device performance~\cite{AppleAPI23}. 
Although a line of work considered wear leveling to limit memory wear~\cite{Chang07,ChangHK07,IraniNR92}, they did not immediately produce high-probability wear guarantees, thus motivating work that focused on hashing algorithms that choose which memory cell to write/overwrite each item, depending on the previous number of writes already incurred by that cell~\cite{EppsteinGMP14}.
Subsequently, the development of specific system software, e.g., the garbage collector, virtual memory manager, or virtual machine hypervisor, automatically handled balancing write operations across memory cells over long time horizons, so that an individual cell would not fail much faster than the overall unit. 
Therefore, subsequent works focused on minimizing the overall number of write operations~\cite{BlellochFGGS15} to the device rather than minimizing the number of write operations to a specific cell. 

\paragraph{Asymmetric read/write costs of large data storage systems.}
Power consumption is a major consideration for large enterprise data storage subsystems, which can often impact the density of servers and the total cost of ownership~\cite{ZhuCTZKW05,ChenYWFSH12}. 
In \cite{BakerHKSO91}, it was observed that given steady hit ratios of read operations, write operations will eventually dominate file system performance and thus be the main source of power consumption. 
Indeed, \cite{NarayananDR08,ThereskaDN11} noted that there are substantial periods of time where all the traffic in the request stream to the Microsoft data storage centers servicing applications, such as Hotmail and Messenger, is write traffic. 

In addition, for distributed data systems that each serve a number of clients, even when one server is updated, they must periodically notify the other servers about their changes to maintain some level of synchronization across the entire system. 
Therefore, write operations to a single server can still induce expensive overheads from communication costs between servers, and thus, reducing the number of writes has long been a goal in disk arrays, distributed systems, and cache-coherent multiprocessors.

\paragraph{Challenges with deterministic algorithms.}
From a theoretical perspective, minimizing the number of internal state changes immediately rules out a large class of deterministic streaming algorithms. 
For example, counting the stream length can be performed deterministically by simply repeatedly updating a counter over the course of a stream. 
Similarly, $L_1$-heavy hitters can be tracked deterministically and in sublinear space using the well-known Misra-Gries data structure~\cite{MisraG82}. 
Another example is the common merge-and-reduce technique for clustering in the streaming model, which can be implemented deterministically if the coreset construction in each reduce step is deterministic. 
Other problems such as maintaining Frequent Directions~\cite{GhashamiLPW16} or $L_2$ regression in the row arrival model~\cite{BravermanDMMUWZ20,BravermanFLRS23} also admit highly non-trivial deterministic algorithms that use sublinear space. 
However, these approaches all update the algorithm upon each stream update and thus seem inherently at odds with achieving a sublinear number of internal state changes over the course of the stream. 

\paragraph{Relationship with sampling.}
On the other hand, sampling-based algorithms inherently seem useful for minimizing the number of internal state changes. 
There are a number of problems that admit sublinear-size coresets based on importance sampling, such as clustering~\cite{BravermanFLR19,BravermanFLSZ21,Cohen-AddadWZ23}, graph sparsification~\cite{AhnG09,BravermanHMSSZ21}, linear regression~\cite{CohenMP20}, $L_p$ regression~\cite{WoodruffY23}, and low-rank approximation~\cite{BravermanDMMUWZ20}. 
These algorithms generally assign some value quantifying the ``importance'' of each stream update as it arrives and then sample the update with probability proportional to the importance. 
Thus if there are few additional operations, then the number of internal state changes can be as small as the overall data structure maintained by the streaming algorithm. 
On the other hand, it is not known that space-optimal sampling algorithms exist for a number of other problems that admit sublinear-space streaming algorithms, such as $F_p$ estimation, $L_p$-heavy hitters, distinct elements, and entropy estimation. 
Hence, a natural question is to ask whether there exist space-optimal streaming algorithms for all of these problems that also incur a small number of internal state changes. 

\subsection{Our Contributions}
In this work, we initiate the study of streaming algorithms that minimize state changes, and demonstrate the existence of algorithms which achieve optimal or near-optimal space bounds while simultaneously achieving an optimal or near optimal number of internal state changes. 

\paragraph{Heavy-hitters.}
We first consider the $L_p$-heavy hitters problem, where the goal is to output estimates $\widehat{f_j}$ to the frequency $f_j$ of every item $j\in[n]$ such that $\left\lvert\widehat{f_j}-f_j\right\rvert\le\frac{\eps}{2}\cdot\|f\|_p$, given an input accuracy parameter $\eps\in(0,1)$. 
Accurate estimation of the heavy-hitter frequencies is important for many other streaming problems, such as moment estimation, $L_p$ sampling~\cite{AndoniKO11,JayaramW18a,JayaramWZ22}, cascaded norms~\cite{JayramW09,JiangLLRW20}, and others. 
Note that under such a guarantee, along with a $2$-approximation of $\|f\|_p$, we can automatically output a list that contains all $j\in[n]$ such that $f_j\ge\eps\cdot\|f\|_p$ but also no index $j\in[n]$ such that $f_j<\frac{\eps}{4}\cdot\|f\|_p$. 
We defer discussion of how to obtain a $2$-approximation to $\|f\|_p$ for the moment and instead focus on the additive error guarantee for all $\widehat{f_j}$. 
Our main result for the heavy hitters problem is the following:
\begin{restatable}{theorem}{thmhhmain}
\thmlab{thm:hh:main}
Given a constant $p\ge 1$, there exists a one-pass insertion-only streaming algorithm that has $\cO{n^{1-1/p}}\cdot\poly\left(\log(nm),\frac{1}{\eps}\right)$ internal state changes, and solves the $L_p$-heavy hitter problem, i.e., it outputs a frequency vector $\widehat{f}$ such that
\[\PPr{\|\widehat{f}-f\|_\infty\le\frac{\eps}{2}\cdot\|f\|_p}\ge\frac{2}{3}.\]
For $p\in[1,2]$, the algorithm uses $\cO{\frac{1}{\eps^{4+4p}}}\cdot\polylog(mn)$ bits of space, while for $p>2$, the algorithm uses $\tO{\frac{1}{\eps^{4+4p}} n^{1-2/p}}$ bits of space. 
\end{restatable}

We next give a lower bound showing that any algorithm solving the $L_p$-heavy hitters problem requires $\Omega(n^{1-1/p})$ state updates.

\begin{restatable}{theorem}{thmhhlb}
\thmlab{thm:hh:lb}
Let $\eps\in(0,1)$ be a constant and $p\ge 1$. 
Any algorithm that solves the $L_p$-heavy hitters problem with threshold $\eps$ with probability at least $\frac{2}{3}$ requires at least $\frac{1}{2\eps}n^{1-1/p}$ state updates.
\end{restatable}
Together, \thmref{thm:hh:main} and \thmref{thm:hh:lb} show that we achieve a near-optimal number of internal state changes. 
Furthermore, \cite{BaIPW10,JowhariST11} showed that for any $p>0$, the $L_p$-heavy hitters problem requires $\Omega\left(\frac{1}{\eps^p}\log n\right)$ words of space, while \cite{Bar-YossefJKS04,Gronemeier09,Jayram09} showed that for $p>2$ and even for constant $\eps>0$, the $L_p$-heavy hitters problem requires $\Omega\left(n^{1-2/p}\right)$ words of space. 
Therefore, \thmref{thm:hh:main} is near-optimal for all $p\ge 1$, for both the number of internal state updates and the memory usage. 

\paragraph{Moment estimation.}
We then consider the $F_p$ moment estimation problem, where the goal is to output an estimate of $F_p(f)=(f_1)^p+\ldots+(f_n)^p$ for a frequency vector $f\in\mathbb{R}^n$ implicitly defined through an insertion-only stream. 
Our main result is the following:
\begin{restatable}{theorem}{thmfpmain}
\thmlab{thm:fp:main}
Given a constant $p\ge 1$, there exists a one-pass insertion-only streaming algorithm that has $\tO{n^{1-1/p}}$ internal state changes, and outputs $\widehat{F_p}$ such that
\[\PPr{\left\lvert\widehat{F_p}-F_p\right\rvert\le\eps\cdot F_p}\ge\frac{2}{3}.\]
For $p\in[1,2]$, the algorithm uses $\cO{\frac{1}{\eps^{4+4p}}}\cdot\polylog(mn)$ bits of space, while for $p>2$, the algorithm uses $\tO{\frac{1}{\eps^{4+4p}} n^{1-2/p}}$ space.
\end{restatable}

We next give a lower bound showing that any approximation algorithm achieving $(2-\Omega(1))$-approximation to $F_p$ requires $\Omega(n^{1-1/p})$ state updates.

\begin{restatable}{theorem}{thmfpconslb}
\thmlab{thm:fp:cons:lb}
Let $\eps\in(0,1)$ be a constant and $p\ge 1$. 
Any algorithm that achieves a $2-\eps$ approximation to $F_p$ with probability at least $\frac{2}{3}$ requires at least $\frac{1}{2}n^{1-1/p}$ state updates.
\end{restatable}
\thmref{thm:fp:cons:lb} shows that our algorithm in \thmref{thm:fp:main} achieves a near-optimal number of internal state changes. 
Moreover, it is known that any one-pass insertion-only streaming algorithm that achieves $(1+\eps)$-approximation to the $F_p$ moment estimation problem requires $\Omega\left(\frac{1}{\eps^2}+\log n\right)$ bits of space~\cite{AlonMS99,Woodruff04} for $p\in[1,2]$ and $\Omega\left(\frac{1}{\eps^2} n^{1-2/p}\right)$ bits of space~\cite{WoodruffZ21} for $p>2$, and thus \thmref{thm:fp:main} is also near-optimal in terms of space for all $p\ge 1$.  

\subsection{Technical Overview}

\paragraph{Heavy-hitters.}
We first describe our algorithm for $L_p$-heavy hitters using near-optimal space and a near-optimal number of internal state changes. 
For ease of discussion, let us assume that $F_p=\widetilde{\Theta}_\eps(n)$, so that the goal becomes to identify the coordinates $j\in[n]$ with $f_j\ge\eps\cdot n^{1/p}$, given an input accuracy parameter $\eps\in(0,1)$. 

We first define a subroutine $\SampleAndHold$ based on sampling a number of items into a reservoir $Q$. 
As we observe updates in the stream, we sometimes update the contents of $Q$ and sometimes observe that some stream updates are to coordinates that are being held by the reservoir. 
For the items that have a large number of stream updates while they are being held by the reservoir, we create separate counters for these items. 

We first describe the intuition for $p>2$. 
We create a reservoir $Q$ of size $\kappa=\widetilde{\calO}_\eps(n^{1-2/p})$ and sample each item of the stream into the reservoir $Q$ with probability roughly $\widetilde{\Theta}_\eps\left(\frac{1}{n^{1/p}}\right)$. 
Note that at some point we may attempt to sample an item of the stream into the reservoir $Q$ when the latter is already full. 
In this case, we choose a uniformly random item of $Q$ to be replaced by the item corresponding to the stream update. 

Our algorithm also checks each stream update to see if it matches an item in the reservoir; if there is a match, we create an explicit counter tracking the frequency of the item. 
In other words, if $j\in[n]$ arrives as a stream update and $j\in Q$ is in the reservoir, then our algorithm $\SampleAndHold$ creates a separate counter for $j$ to count the number of subsequent instances of $j$. 

Now for a heavy hitter $j\in[n]$, we have $f_j\ge\eps\cdot n^{1/p}$ and thus since the sampling probability is $\widetilde{\Theta}_\eps\left(\frac{1}{n^{1/p}}\right)$, then we can show that $j$ will likely be sampled into our reservoir $Q$ at some point. 
In fact, since the reservoir $Q$ has size $\kappa=\widetilde{\calO}_\eps(n^{1-2/p})$, then in expectation, $j$ will be retained by the reservoir for roughly $\widetilde{\Omega}_\eps(n^{1-1/p})$ stream updates before it is possibly replaced by some other item. 
Moreover, since $f_j\ge\eps\cdot n^{1/p}$, then we should expect another instance of $j$ to arrive in $\widetilde{\Omega}_\eps(n^{1-1/p})$ additional stream updates, where the expectation is taken over the randomness of the sampled positions, under the assumption that $F_p=\Theta(n)$, which implies also that the stream length is at most $\cO{n}$. 
Therefore, our algorithm will create a counter for tracking $j$ before $j$ is removed from the reservoir $Q$. 
Furthermore, we can show that the counter for $j$ is likely created sufficiently early in the stream to provide a $(1+\eps)$-approximation to the frequency $f_j$ of $j$. 
Then to decrease the number of internal state updates, we can use Morris counters to approximate the frequency of subsequent updates for each tracked item. 

\paragraph{Counter maintenance for heavy-hitters.}
The main issue with this approach is that too many counters can be created. 
As a simple example, consider when all items of each coordinate arrive together, one item after the other.  
Although this is a simple case for counting heavy-hitters, our algorithm will create counters for almost any item that it samples, and although our reservoir uses space $\widetilde{\calO}_\eps(n^{1-2/p})$, the total number of items sampled over the course of the stream is $\widetilde{\calO}_\eps(n^{1-1/p})$ and thus the number of created counters can also be $\widetilde{\calO}_\eps(n^{1-1/p})$, which would be too much space. 
We thus create an additional mechanism to remove half of the counters each time the number of counters becomes too large, i.e., exceeds $\widetilde{\calO}_\eps(n^{1-2/p})$. 

The natural approach would be to remove half of the counters with the smallest tracked frequencies. 
However, one could imagine a setting where $m=\Theta(n)$ and the $p$-th moment of the stream is also $\Theta(n)$, but there exists a heavy-hitter with frequency $\Theta(n^{1/p})$ that appears every $\Theta(n^{1-1/p})$ updates. 
In this case, even if the heavy-hitter is sampled, its corresponding counter will be too low to overcome the other counters in the counter maintenance stage. 
Therefore, we instead remove the counters with the smallest tracked frequencies for each set of counters that have been maintained for a number of steps between $2^z$ and $2^{z+1}$ for each integer $z\ge 0$, which overcomes per-counter analysis in similar algorithms based on sampling~\cite{BravermanO13,BravermanKSV14}. 

Since our algorithm samples each item of the stream of length $\cO{n}$ with probability $\frac{1}{\widetilde{\calO}_\eps(n^{1/p})}$, then we expect our reservoir to have $\widetilde{\Omega}_\eps(n^{1-1/p})$ internal state changes. 
On the other hand, the counters can increment each time another instance of the tracked item arrives. 
To that end, we replace each counter with an approximate counter that has a small number of internal state changes. 
In particular, by using Morris counters as mentioned above, the number of internal state changes for each counter is $\poly\left(\log n,\frac{1}{\eps},\log\frac{1}{\delta}\right)$ times over the course of the stream. 
Therefore, the total number of internal state changes is $\widetilde{\Omega}_\eps(n^{1-1/p})$ while the total space used is $\widetilde{\calO}_\eps(n^{1-2/p})$. 

For $p\in[1,2]$, we instead give the reservoir $Q$ a total of $\kappa=\poly_{\eps}(\log n)$ size, so that the total space is $\poly_{\eps}(\log n)$ while the total number of internal state changes remains $\widetilde{\Omega}_\eps(n^{1-1/p})$. 

\paragraph{Removing moment assumptions.}
To remove the assumption that $F_p=\widetilde{\calO}_\eps(n)$, we note that if each element of the stream of length $m$ is sampled with probability $q<1$, then the expected number of sampled items is $qm$, but the $p$-th power of the expected number of sampled items is $(qm)^p$. 
Although this is not the $p$-th moment of the stream, we nevertheless can expect the $F_p$ moment of the stream to decrease at a faster rate than the number of sampled items. 
Thus we create $L=\cO{\log(nm)}$ substreams so that for each $\ell\in[L]$, we subsample each stream update $[m]$ with probability $\frac{1}{2^{\ell-1}}$. 
For one of these substreams $J_\ell$, we will have $F_p(J_\ell)=\widetilde{\calO}_\eps(n)$. 
We show that we can estimate the frequency of the heavy-hitters in the substream $J_\ell$ and then rescale by the inverse sampling rate to achieve a $(1+\eps)$-approximation to the frequency of the heavy-hitters in the original stream. 

It then remains to identify the correct stream $\ell$ such that $F_p(J_\ell)=\widetilde{\calO}_\eps(n)$. 
A natural approach would be to approximate the moment of each substream, to identify such a correct stream. 
However, it turns out that our $F_p$ moment estimation algorithm will ultimately use our heavy hitter algorithm as a subroutine. 
Furthermore, other $F_p$ moment estimation algorithms, e.g.,~\cite{AlonMS99,Indyk06,Li08,GangulyW18}, use a number of internal state changes that is linear in the stream length and it is unclear how to adapt these algorithms to decrease the number of internal state changes. 
Instead, we note that with high probability, the estimated frequency of each heavy-hitter by our algorithm can only be an underestimate. 
This is because if we initialize counters throughout the stream to track the heavy hitters, then our counters might miss some stream updates to the heavy hitters, but it is not possible to overcount the frequency of each heavy hitter, i.e., we cannot count stream updates that do not exist. 
Moreover, this statement is still true, up to a $(1+\eps)$ factor, when we use approximate counters. 
Therefore, it suffices to use the maximum estimation for the frequency for each heavy hitter, across all the substreams. 
We can then use standard probability boosting techniques to simultaneously accurately estimate all $L_p$-heavy hitters. 

\paragraph{Moment estimation.}
Given our algorithm for finding $(1+\eps)$-approximations to the frequencies of $L_p$-heavy hitters, we now adapt a standard subsampling framework~\cite{IndykW05} to reduce the $F_p$ approximation problem to the problem of finding the $L_p$-heavy hitters. 
The framework has subsequently been used in a number of different applications, e.g.,~\cite{WoodruffZ12,BlasiokBCKY17,LevinSW18,BravermanWZ21,WoodruffZ21,MahabadiWZ22,BravermanMWZ23} and has the following intuition. 

For ease of discussion, consider the level set $\Gamma_i$ consisting of the coordinates $j\in[n]$ such that $f_j\in\left(\frac{F_p}{2^i},\frac{F_p}{2^{i+1}}\right]$ for each $i$, though we remark that for technical reasons, we shall ultimately define the level sets in a slightly different manner. 
Because the level sets partition the universe $[n]$, then if we define the contribution $C_i:=\sum_{j\in\Gamma_i}(f_k)^p$ of a level set $\Gamma_i$ to be the sum of the contributions of all their coordinates, then we can decompose the moment $F_p$ into the sum of the contributions of the level sets, $F_p=\sum_i C_i$. 
Moreover, it suffices to accurately estimate the contributions of the ``significant'' level sets, i.e., the level sets whose contribution is at least a $\poly\left(\eps,\frac{1}{\log(nm)}\right)$ fraction of the $F_p$ moment, and crudely estimate the contributions of the insignificant level sets. 

\cite{IndykW05} observed that the contributions of the significant level sets can be estimated by approximating the frequencies of the heavy hitters for substreams induced by subsampling the universe at exponentially smaller rates. 
We emphasize that whereas we previously subsampled updates of the stream $[m]$ for heavy hitters, we now subsample elements of the universe $[n]$. 
That is, we create $L=\cO{\log(nm)}$ substreams so that for each $\ell\in[L]$, we subsample each element of the universe $[n]$ into the substream with probability $\frac{1}{2^{\ell-1}}$. 
For example, a single item with frequency $F_p^{1/p}$ will be a heavy hitter in the original stream, which is also the stream induced by $\ell=1$. 
On the other hand, if there are $n$ items with frequency $(F_p/n)^{1/p}$, then they will be $L_p$-heavy hitters at a subsampling level where in expectation, there are roughly $\Theta\left(\frac{1}{\eps^p}\right)$ coordinates of the universe that survive the subsampling. 
Then \cite{IndykW05} notes that $(1+\eps)$-approximations to the contributions of the surviving heavy-hitters can then be rescaled inversely by the sampling rate to obtain good approximations of the contributions of each significant level set. 
The same procedure also achieves crude approximations for the contributions of insignificant level sets, which overall suffices for a $(1+\eps)$-approximation to the $F_p$ moment. 

The key advantage in adapting this framework over other $F_p$ estimation algorithms, e.g.,~\cite{AlonMS99,Indyk06,Li08,GangulyW18} is that we can then use our heavy hitter algorithm $\FullSampleAndHold$ to provide $(1+\eps)$-approximations to the heavy hitters in each substream while guaranteeing a small number of internal state changes.   

\subsection{Additional Intuition and Comparison with Previous Algorithms}
There are a number of differences between our algorithm and the sample-and-hold approach of \cite{EstanV02}. Firstly, once \cite{EstanV02} samples an item, a counter will be initialized and maintained indefinitely for that item. By comparison, our algorithm will sample more items than the total space allocated to the algorithm, so we must carefully delete a number of sampled items. In particular, it is NOT correct to delete the sampled items with the largest counter. Secondly, \cite{EstanV02} updates a counter each time a subsequent instance of the sample arrives. Because our paper is focused on a small number of internal state changes, our algorithm cannot afford such a large number of updates. Instead, we maintain approximate counters that sacrifice accuracy in exchange for a smaller number of internal state changes. We show that the loss in accuracy can be tolerated in choosing which samples to delete. 

Another possible point of comparison is the precision sampling technique of \cite{AndoniKO11}, which is a linear sketch, so although it has an advantage of being able to handle insertion-deletion streams, unfortunately it must also be updated for each stream element arrival, resulting in a linear number of internal state changes. Similarly, a number of popular heavy-hitter algorithms such as Misra-Gries~\cite{MisraG82}, CountMin~\cite{CormodeM05}, CountSketch~\cite{CharikarCF04}, and SpaceSaving~\cite{MetwallyAA05} can only achieve a linear number of internal state changes. By comparison, our sample-and-hold approach results in a sublinear number of internal state changes. 

Finally, several previous algorithms are also based on sampling a number of items throughout the stream, temporarily maintaining counters for those items, and then only keeping the items that are globally heavy, e.g., \cite{BravermanO13,BravermanKSV14}. It is known that these algorithms suffer a bottleneck at $p=3$, i.e., they cannot identify the $L_p$ heavy-hitters for $p<3$. The following counterexample shows why these algorithms cannot identify the $L_2$ heavy-hitters and illustrates a fundamental difference between our algorithms. 

Suppose the stream consists of $\sqrt{n}$ blocks of $\sqrt{n}$ updates. Among these updates, there are $\sqrt{n}$ items with frequency $n^{1/4}$, which we call pseudo-heavy. 
There is a single item with frequency $\sqrt{n}$, which is the heavy-hitter. 
Then the remaining items each have frequency $1$ and are called light. 
Note that the second moment of the stream is $\Theta(n)$, so that only the item with frequency $\sqrt{n}$ is the heavy-hitter, for constant $\eps<1$. 

Let $S=\{1,2,\ldots,n^{1/4}\}$ and suppose for each $w\in S$, block $w$ is a special block that consists of $n^{1/4}$ different pseudo-heavy items, each with frequency $n^{1/4}$. 
Let $T=x+S$, for $x=\{1,2,\ldots,n^{1/8}\}$, so that $T$ consists of the $n^{1/8}$ blocks after each special block. 
Each block in $T$ consists of $n^{1/8}$ instances of the heavy-hitter, along with $\sqrt{n}-n^{1/8}$ light items. 
The remaining blocks all consist of light items. 

Observe that without dynamic maintenance of counters for different scales, in each special block, we will sample $\polylog(n)$ pseudo-heavy items whose counters each reach about $\tO{n^{1/4}}$. But then each time a heavy-hitter is sampled, its count will not exceed the pseudo-heavy item before the number of counters before it is deleted, because it only has $n^{1/8}$ instances in its block. Thus with high probability, the heavy-hitter will never be found, and this is an issue with previously existing sampling-based algorithms, e.g., \cite{BravermanO13,BravermanKSV14}. 

Our algorithm overcomes this challenge by only performing maintenance on counters that have been initialized for a similar amount of time. Thus in the previous example, the counters for the heavy-hitters will not be deleted because they are not compared to the counters for the pseudo-heavy items until the heavy-hitters have sufficiently high frequency. By comparison, existing algorithms will retain counters for the pseudo-heavy items, because they locally look ``larger'', at the expense of the true heavy-hitter.

\begin{table}[!htb]
\centering
\begin{tabular}{|c|c|c|}\hline
Reference & State Changes & Setting \\\hline
\cite{MisraG82} & $\cO{m}$ & $L_1$-Heavy Hitters Only \\\hline
\cite{CormodeM05} & $\cO{m}$ & $L_1$-Heavy Hitters Only \\\hline
\cite{MetwallyAA05} & $\cO{m}$ & $L_1$-Heavy Hitters Only \\\hline
\cite{CharikarCF04} & $\cO{m}$ & $L_2$-Heavy Hitters \\\hline
Our Work & $\tO{n^{1-1/p}}$ & $L_2$-Heavy Hitters \\\hline
\end{tabular}
\caption{Summary of our results compared to existing results for a stream of length $m$ on a universe of size $n$. We emphasize that reporting $L_2$ heavy-hitters includes the $L_1$ heavy-hitters. All algorithms use near-optimal space.}
\tablelab{table:summary}
\end{table}

\subsection{Preliminaries}
We use $[n]$ to denote the set $\{1,2,\ldots,n\}$ for an integer $n>0$. 
We write $\poly(n)$ to denote a fixed univariate polynomial in $n$ and similarly, $\poly(n_1,\ldots,n_k)$ to denote a fixed multivariate polynomial in $n_1,\ldots,n_k$. 
We use $\tO{f(n_1,\ldots,n_k)}$ for a function $f(n_1,\ldots,n_k)$ to denote $f(n_1,\ldots,n_k)\cdot\poly(\log f(n_1,\ldots,n_k))$. 
For a vector $f\in\mathbb{R}^n$, we use $f_i$ for $i\in[n]$ to denote the $i$-th coordinate of $f$. 

\paragraph{Model.} 
In our setting, an insertion-only stream $S$ consists of a sequence of updates $u_1,\ldots,u_m$. 
In general, we do not require $m$ to be known in advance, though in some cases, we achieve algorithmic improvements when a constant-factor upper bound on $m$ is known in advance; we explicitly clarify the setting in these cases. 
For each $t\in[m]$, we have $u_t\in[n]$, where without loss of generality, we use $[n]$ to represent an upper bound on the universe, which we assume to be known in advance. 
The stream $S$ defines a frequency vector $f\in\mathbb{R}^n$ by $f_i=|\{t\,\mid\,u_t=i\}|$, so that for each $i\in[n]$, the $i$-th value of the frequency vector $S$ is how often $i$ appears in the data stream $S$. 
Observe that through a simple reduction, this model suffices to capture the more general case where some coordinate is increased by some amount in $\{1,\ldots,M\}$ for some integer $M>0$ in each update. 

For a fixed algorithm $\calA$, suppose that at each time $t\in[m]$,  the algorithm $\calA$ maintains a memory state $\sigma_t$. 
Let $|\sigma_t|$ denote the size of the memory state, in words of space, where each word of space is assumed to used $\cO{\log n+\log m}$ bits of space. 
We say the algorithm $\calA$ uses $s$ words of space or equivalently, $\cO{s\log n+\log m}$ bits of space, if $\max_{t\in[m]}|\sigma_t|\le s$. 
For each $t\in[m]$, let $X_t$ be the indicator random variable for whether the algorithm $\calA$ changed its memory state at time $t$. 
That is, $X_t=1$ if $\sigma_t\neq\sigma_{t-1}$ and $X_t=0$ if $\sigma_t=\sigma_{t-1}$, where we use the convention $\sigma_0=\emptyset$. 
Then we say the total number of internal changes by the algorithm $\calA$ is $\sum_{t=1}^m x_t$. 

We also require the following definition of Morris counters to provide approximate counting. 
\begin{theorem}[Morris counters]
\cite{Morris78,NelsonY22}
\thmlab{thm:morris}
There exists an insertion-only streaming algorithm (Morris counter) that uses space (in bits) $\cO{\log\log n+\log\frac{1}{\eps}+\log\log\frac{1}{\delta}}$ and outputs a $(1+\eps)$-approximation to the frequency of an item $i$, with probability at least $1-\delta$. 
Moreover, the algorithm is updated at most $\poly\left(\log n,\frac{1}{\eps},\log\frac{1}{\delta}\right)$ times over the course of the stream. 
\end{theorem}

\section{Heavy Hitters}
\seclab{sec:hh}
In this section, we first describe our algorithm for identifying and accurately approximating the frequencies of the $L_p$-heavy hitters. 

\subsection{Sample and Hold}
A crucial subroutine for our $F_p$ estimation algorithm is the accurate estimation of heavy hitters. 
In this section, we first describe such a subroutine $\SampleAndHold$ for approximating the frequencies of the $L_p$-heavy hitters under the assumption that $F_p$ is not too large. 

We now describe our algorithm for the case $p\ge 2$. 
Our algorithm creates a reservoir $Q$ of size $\kappa=\widetilde{\calO}_\eps(n^{1-2/p})$ and samples each item of the stream into the reservoir $Q$ with probability roughly $\widetilde{\Theta}_\eps\left(\frac{1}{n^{1/p}}\right)$. 
If the reservoir $Q$ is full when an item of the stream is sampled, then a uniformly random item of $Q$ is replaced with the stream update. 
Thus if the stream has length $n$, then we will incur $\widetilde{\calO}_\eps(n^{1-2/p})$ internal state changes due to the sampling. 
For stream length $m$, we set the sampling probability to be roughly $\frac{\widetilde{\calO}_\eps(n^{1-1/p})}{m}$. 

Our algorithm also checks each stream update to see if it matches an item in the reservoir and creates a counter for the item if there is a match. 
In other words, if $j\in[n]$ arrives as a stream update and $j\in Q$ is in the reservoir, then our algorithm $\SampleAndHold$ creates a separate counter for $j$ to count the number of subsequent instances of $j$. 
In addition, each time the number of counters becomes too large, i.e., exceeds $\widetilde{\calO}_\eps(n^{1-2/p})$, we remove half of the counters that have been maintained for a time between $2^z$ and $2^{z+1}$, for each integer $z>0$. 
In particular, we remove the counters with the smallest tracked frequencies for each of the groups. 
To reduce the number of internal state changes, we use Morris counters rather than exact counters for each item. 

For $p\in[1,2]$, we set $\kappa=\poly_{\eps}(\log n)$ to be the size of the reservoir $Q$, so that the total space is $\poly_{\eps}(\log n)$ while the total number of internal state changes remains $\widetilde{\Omega}_\eps(n^{1-1/p})$. 
Our algorithm $\SampleAndHold$ appears in \algref{alg:sample:hold}. 

\begin{algorithm}[!htb]
\caption{$\SampleAndHold$}
\alglab{alg:sample:hold}
\begin{algorithmic}[1]
\Require{Stream $s_1,\ldots,s_m$ of items from $[n]$, accuracy parameter $\eps\in(0,1)$, $p\ge 1$}
\Ensure{Accurate estimation of an $L_p$ heavy hitter frequency}
\State{$\kappa_1\gets\Theta\left(\frac{\log^{11+3p}(mn)}{\eps^{4+4p}}\right)$, $\gamma\gets2^{20}p$}
\If{$m\ge n$}
\State{$\varrho\gets\frac{\gamma^2n^{1-1/p}\log^4(nm)}{\eps^2 m}$, $\kappa_2\gets\Theta\left(\frac{n^{1-2/p}\log^{11+3p}(mn)}{\eps^{4+4p}}\right)$}
\ElsIf{$m<n$}
\State{$\varrho\gets\frac{\gamma^2m^{1-1/p}\log^4(nm)}{\eps^2 m}$, $\kappa_2\gets\Theta\left(\frac{m^{1-2/p}\log^{11+3p}(mn)}{\eps^{4+4p}}\right)$}
\EndIf
\State{$\kappa\gets\kappa_1$ if $p\in[1,2)$, $\kappa\gets\kappa_2$ if $p\ge2$}
\State{$k\sim\Uni([200p\kappa\log^2(nm),202p\kappa\log^2(nm)])$}
\State{$q_i\gets\emptyset$ for $i\in[k]$}
\For{$t=1$ to $t=m$}
\If{there is a Morris counter for $s_t$}
\State{Update the Morris counter}
\ElsIf{there exists $i\in[k]$ with $q_i=s_t$}
\Comment{item is in the reservoir}
\State{Start a Morris counter for $s_t$}
\Comment{hold a counter for the item}
\Else
\State{Pick $\mu_t\in[0,1]$ uniformly at random}
\If{$\mu_t<\varrho$}
\Comment{with probability $\varrho$}
\State{Pick $i\in[k]$ uniformly at random}
\State{$q_i\gets s_t$}
\EndIf
\EndIf
\If{there exist $k$ active Morris counters initialized between time $t-2^z$ and $t-2^{z+1}$ for integer $z>0$}
\Comment{too many counters}
\State{$k\gets\Uni([200p\kappa\log^2(nm),202p\kappa\log^2(nm)])$}
\State{Retain the $\frac{k}{2}$ counters initialized between time $t-2^z$ and $t-2^{z+1}$ for integer $z>0$ with largest approximate frequency}
\EndIf
\EndFor
\State{\Return the estimated frequencies by the Morris counters}
\end{algorithmic}
\end{algorithm}

We first analyze how many additional counters are available at a time when a heavy hitter $j$ is sampled. 
\begin{restatable}{lemma}{lemrandomcounters}
\lemlab{lem:random:counters}
Let $j\in[n]$ be an item with $(f_j)^p\ge\frac{\eps^2\cdot F_p}{2^{10}\gamma\log^2(nm)}$ and let $J=\{t\in[m]\,\mid\, s_t=j\}$. 
Let $k\in[200p\kappa\log^2(nm),202p\kappa\log^2(nm)]$ be chosen uniformly at random. 
Let $v$ be the last time that $j$ is sampled by the algorithm. 
Then with probability at least $1-\frac{1}{50p\log^2(nm)}$, there are at most $k-\kappa$ counters at time $v$.  
\end{restatable}
\begin{proof}
Consider any fixing of the random samples of the algorithm and the random choices of $k$, before time $v$. 
Let $T_1<T_2<\ldots$ be the sequence of times when the counters are reset. 
Note that since $k\in[200p\kappa\log^2(nm),202p\kappa\log^2(nm)]$, then each time the counters are reset, between $[100p\kappa^2(nm),101p\kappa\log^2(nm)]$ counters are newly allocated. 
Note that the sequence could be empty, in which case our claim is vacuously true. 
For each $T_i$, consider the times $\calT_i$ at which additional counters would be created after $T_i$ is fixed if there were no limit to the number of counters. 
Moreover, let $u_i$ be the choice of $k$ at time $T_i$. 
Note that the first $100p\kappa\log^2(nm)$ of the times in $\calT_i$ are independent of the choice of $u_i$, while latter times in $\calT_i$ may not actually be sampled due to the choice of $u_i$. 
Let $T_w$ be the first time for which which $v$ appears in the first $101p\kappa\log^2(nm)$ terms of the $\calT_w$. 
Then with probability at most $\frac{1}{100p\log^2(mn)}$, the choice of $u_w$ will be within $\kappa$ indices after $v$ in the sequence $\calT_i$. 
On the other hand, the choice of $u_w$ could cause $T_{w+1}$ to be before $v$, e.g., if $v$ is the $(101p\kappa\log^2(nm))$-th term and $u_w=100\kappa\log^2(nm)$, in which case the same argument shows that with probability at most $\frac{1}{100p\log^2(mn)}$, the choice of $u_{w+1}$ will be within $\kappa$ indices after $v$ in the sequence $\calT_{w+1}$. 
Note that since $u_w+u_{w+1}\ge200p\kappa^2(nm)$, then $v$ must appear before $T_{w+2}$. 
Therefore by a union bound, with probability at least $1-\frac{1}{50p\log^2(nm)}$, there are at most $k-\kappa$ counters at time $v$.  
\end{proof}
We next upper bound how many additional counters are created between the time a heavy hitter $j$ is sampled until it becomes too large to delete. 
\begin{restatable}{lemma}{lemnumnewcounters}
\lemlab{lem:num:new:counters}
Suppose $m\ge n$ and $F_p=\cO{\frac{n\log^3(nm)}{\eps^{4p}}}$. 
Let $j\in[n]$ be an item with $(f_j)^p\ge\frac{\eps^2\cdot F_p}{2^{10}\gamma\log^2(nm)}$. 
Let $J=\{t\in[m]\,\mid\, s_t=j\}$ and let $u\in J$ be chosen uniformly at random and suppose that the algorithm samples $j$ at time $u$. 
Then for $p\in[1,2)$, over the choice of $u$ and the internal randomness of the algorithm, the probability that fewer than $\kappa_1=\cO{\frac{\log^{11+3p}(mn)}{\eps^{4+4p}}}$ new counters are generated after $u$ and before $\frac{\eps^4\cdot F_p}{2^{10}\gamma\log^2(nm)}$ additional instances of $j$ arrive is at most $1-\frac{1}{100p\log(nm)}$. 

Similarly for $p\ge2$, over the choice of $u$ and the internal randomness of the algorithm, the probability that fewer than $\kappa_2=\cO{\frac{n^{1-2/p}\log^{11+3p}(mn)}{\eps^{4+4p}}}$ new counters are generated after $u$ and before $\frac{\eps^4\cdot F_p}{2^{10}\gamma\log^2(nm)}$ additional instances of $j$ arrive is at most $1-\frac{1}{100p\log(nm)}$. 
\end{restatable}
\begin{proof}
Let $L=\cO{p\log(nm)}$ and let $\ell\in[L]$ be fixed. 
Let $W_\ell$ be the items with frequency in $[2^{\ell-1},2^\ell)$ so that
\[W_\ell=\{i\in[n]\,\mid\,f_i\in[2^{\ell-1},2^\ell)\}.\]
and observe that $|W_\ell|\le\frac{F_p}{2^{p\ell}}$. 
Let $\varrho=\frac{\gamma^2n^{1-1/p}\log^4(nm)}{\eps^2 m}$ and $X=\frac{\eps^2 F_p^{1/p}}{2^{14}\gamma^2\log^4(nm)}$. 

We define $B_1,\ldots,B_\alpha$ to be blocks that partition the stream, so that the $i$-th block includes the items of the stream after the $(i-1)$-th instance of $f_j$, up to and including the $i$-th instance of $f_j$. 
We therefore have $\alpha=f_j+1$. 

Observe that since $(f_j)^p\ge\frac{\eps^2\cdot F_p}{2^{10}\gamma\log^2(nm)}$, then we have $f_j=\beta X$ for some $\beta>1$. 
%
Since $|W_\ell|\le\frac{F_p}{2^{p\ell}}$, then the expected number of unique items in $W_\ell$ contained in a block is at most $\frac{F_p}{2^{p\ell}\alpha}$, conditioned on any fixing of indices that are sampled. 
Thus in a block, the conditional expectation of the number of stream updates that correspond to items in $W_\ell$ is at most $\frac{F_p}{2^{(p-1)\ell}\alpha}$, and so the expected number of items in $W_\ell$ that are sampled in a block is at most $\frac{\varrho F_p}{2^{(p-1)\ell}\alpha}$. 
Therefore, the expected number of retained items in the previous and following $2^i$ blocks for $i\le\ell$ is at most 
\[\frac{\varrho F_p}{2^{(p-1)\ell}\alpha}\cdot\min(2^i,2^\ell)\le\frac{\varrho F_p}{2^{(p-2)\ell}\alpha}\le\frac{2\varrho F_p}{2^{(p-2)\ell}X},\]
since $\alpha=f_j+1$ and by assumption $(f_j)^p\ge\frac{\eps^2\cdot F_p}{2^{10}\gamma\log^2(nm)}$, but $X=\frac{\eps^2 F_p^{1/p}}{2^{14}\gamma^2\log^4(nm)}$. 
For $i>\ell$, note that since $W_\ell$ only contains elements of frequency $2^\ell$, then no elements of $W_\ell$ will be retained over $j$ once the consideration is for $2^i>2^\ell$ blocks. 

For $p\in[1,2)$, note that $2^\ell\le F_p^{1/p}$, so that $\frac{1}{2^{(p-2)\ell}}=2^{(2-p)\ell}\le F_p^{2/p-1}$. 
Therefore, after $2^i$ blocks, the expected number of items in $W_\ell$ is at most $\frac{2\varrho F_p}{2^{(p-2)\ell}X}\le\frac{2\varrho F_p^{2/p}}{X}$. 
For $\varrho=\frac{\gamma^2n^{1-1/p}\log^4(nm)}{\eps^2 m}$ and $X=\frac{\eps^2 F_p^{1/p}}{2^{14}\gamma^2\log^4(nm)}$, we have that the expected number of retained items in the previous and following $2^i$ blocks is at most
\[\cO{\frac{2\gamma^4 n^{1-1/p}F_p^{1/p}\log^8(nm)}{\eps^4 m}}=\cO{\frac{\log^{8+3p}(nm)}{\eps^{4+4p}}},\]
for $m\ge n$ and $F_p=\cO{\frac{n\log^{3p}(nm)}{\eps^{4p}}}$.
By Markov's inequality, we have that with probability $1-\frac{1}{100p\log^3(nm)}$ over a random time $u$, the number of counters for the items with frequency in $[2^{\ell-1},2^\ell)$ is at most $\kappa_1=\cO{\frac{\log^{11+3p}(nm)}{\eps^{4+4p}}}$ across $2^i$ blocks before and after $u$. 
Thus, by a union bound over all $i=\cO{\log m}$ and $L=\cO{p\log(nm)}$ choices of $\ell$, we have that with probability at least $1-\frac{1}{100p\log(nm)}$, $\kappa_1$ new counters are generated after $u$.

For $p\ge 2$, we have that the expected number of retained items across $2^i$ blocks is at most $\frac{2\varrho F_p}{2^{(p-2)\ell}X}\le\frac{2\varrho F_p}{X}$. 
For $\varrho=\frac{\gamma^2n^{1-1/p}\log^4(nm)}{\eps^2 m}$ and $X=\frac{\eps^2 F_p^{1/p}}{2^{14}\gamma^2\log^4(nm)}$, we have that the expected number of retained items across $2^i$ blocks is at most
\[\cO{\frac{2\gamma^4 (nF_p)^{1-1/p}\log^8(nm)}{\eps^4 m}}=\cO{\frac{n^{1-2/p}\log^{8+3p}(mn)}{\eps^{4+4p}}},\]
for $m\ge n$ and $F_p=\cO{\frac{n\log^{3p}(nm)}{\eps^{4p}}}$. 
By Markov's inequality, we have that with probability $1-\frac{1}{100p\log^3(nm)}$ over a random time $u$, the number of counters for the items with frequency in $[2^{\ell-1},2^\ell)$ is at most $\kappa_2=\cO{\frac{n^{1-2/p}\log^{11+3p}(mn)}{\eps^{4+4p}}}$ across $2^i$ blocks after $u$. 
Moreover, observe that the number of counters generated within $2^z$ timesteps is certainly at most the number of counters generated within $2^z$ blocks. 
Thus, by a union bound over all $z=\cO{\log m}$ and $L=\cO{p\log(nm)}$ choices of $\ell$, we have that with probability at least $1-\frac{1}{100p\log(nm)}$, $\kappa_2$ new counters are generated after $u$. 
\end{proof}

By \lemref{lem:random:counters} and \lemref{lem:num:new:counters}, we have:
\begin{lemma}
\lemlab{lem:counters:bound}
Suppose $m\ge n$ and $F_p=\cO{\frac{n\log^3(nm)}{\eps^{4p}}}$. 
Let $j\in[n]$ be an item with $(f_j)^p\ge\frac{\eps^2\cdot F_p}{2^{10}\gamma\log^2(nm)}$. 
Then with probability at least $1-\frac{1}{100p\log(nm)}$, the counters are not reset between the times at which $j$ is sampled and $\frac{\eps^4\cdot F_p}{2^{10}\gamma\log^2(nm)}$ occurrences of $j$ arrive after it is sampled. 
\end{lemma}

We now claim that a heavy hitter $j$ will be sampled early enough to obtain a good approximation to its overall frequency. 
\begin{restatable}{lemma}{lemonesamplehold}
\lemlab{lem:one:sample:hold}
Suppose $F_p=\cO{\frac{n\log^3(nm)}{\eps^{4p}}}$. 
Let $j\in[n]$ be an item with $(f_j)^p\ge\frac{\eps^2\cdot F_p}{2^{10}\gamma\log^2(nm)}$. 
Then with probability at least $0.99$, $\SampleAndHold$ outputs $\widehat{f_j}$ such that
\[\left(1-\frac{\eps}{8\log(nm)}\right)\cdot(f_j)^p\le(\widetilde{f_j})^p\le\left(1+\frac{\eps}{8\log(nm)}\right)\cdot(f_j)^p.\]
\end{restatable}
\begin{proof}
Note that for the purposes of analysis, we can assume $m\ge n$, since otherwise if $m<n$, then $\SampleAndHold$ essentially redefines $n$ to be the number of unique items in the induced stream by setting $\varrho$ and $\kappa$ appropriately, even though the overall universe can be larger. 
Note that $m\le2\left(F_p\right)^{1/p}\cdot n^{1-1/p}$. 

By assumption, we have $(f_j)^p\ge\frac{\eps^2\cdot F_p}{2^7\gamma\log^2(nm)}$ and thus, we certainly have $f_j\ge\frac{\eps\cdot\left(F_p\right)^{1/p}}{2^7\gamma\log^2(nm)}$ for $p\ge 2$ and $\eps\in(0,1)$. 
Let $T$ be the set of the first $\frac{\eps}{16\log(nm)}$ fraction of occurrences of $j$, so that 
\[|T|\ge\frac{\eps^2\cdot\left(F_p\right)^{1/p}}{2^{11}\gamma\log^3(nm)}.\]
We claim that with probability at least $\frac{2}{3}$, a Morris counter for $j$ will be created by the stream as it passes through $T$. 
Indeed, observe that since each item of the stream is sampled with probability 
\[\varrho=\frac{\gamma^2n^{1-1/p}\log^4(nm)}{\eps^2 m}\ge\frac{\gamma^2\log^4(nm)}{2\eps^2\left(F_p\right)^{1/p}},\]
then we have with high probability, an index in $T$ is sampled. 
By \lemref{lem:counters:bound}, the index will not be removed by the counters resetting. 

Since $T$ is the set of the first $\frac{\eps}{16\log(nm)}$ fraction of occurrences of $j$, then the Morris counter is used for at least $\left(1-\frac{\eps}{16\log(nm)}\right)f_j$ occurrences of $j$. 
We use Morris counters with multiplicative accuracy $\left(1+\cO{\frac{\eps}{\log(nm)}}\right)$. 
Hence by \thmref{thm:morris}, we obtain an output $\widehat{f_j}$ such that 
\[\left(1-\frac{\eps}{8\log(nm)}\right)\cdot(f_j)^p\le(\widetilde{f_j})^p\le\left(1+\frac{\eps}{8\log(nm)}\right)\cdot(f_j)^p.\]
\end{proof}
Finally, for the purposes of completeness, we remark that an item that is not a heavy-hitter cannot be reported by our algorithm because our algorithm only counts the number of instances of tracked items and thus it always reports underestimates of each item. 
Underestimates of items that are not heavy-hitters will be far from the threshold for heavy-hitters and thus, those items will not be reported. 

\subsection{Full Sample and Hold}
We now address certain shortcomings of \lemref{lem:one:sample:hold} -- namely, the assumption that $F_p=\cO{\frac{n\log^3(nm)}{\eps^{4p}}}$ and the fact that \lemref{lem:one:sample:hold} only provides constant success probability for each heavy hitter $j\in[n]$, but there can be many of these heavy-hitters. 

\begin{algorithm}[!htb]
\caption{$\FullSampleAndHold$}
\alglab{alg:high:sample:hold}
\begin{algorithmic}[1]
\Require{Stream $s_1,\ldots,s_m$ of items from $[n]$, accuracy parameter $\eps\in(0,1)$, $p\ge 1$}
\Ensure{Accurate estimations of $L_p$ heavy hitter frequencies}
\State{$R\gets\cO{\log n}$, $Y=\cO{\log m}$}
\For{$r\in[R], x\in[Y]$}
\State{Let $J^{(r)}_x$ be a (nested) subset of $[m]$ subsampled at rate $p_x:=\min(1, 2^{1-x})$}
\State{Let $m^{(r)}_x$ be the length of the stream $J^{(r)}_x$}
\State{Run $\SampleAndHold^{(r)}_x$ on $J^{(r)}_x$}
\EndFor
\State{Let $\widehat{f_j^{(r,x)}}$ be the estimated frequency for $j$ by $\SampleAndHold^{(r)}_x$}
\State{$\widehat{f_j^{x}}\gets\median_{r\in[R]}\widehat{f_j^{(r,x)}}$}
\State{Let $\ell=\min\{x\in[X]\,\mid\,m_x\ge(\widehat{f_j^{x}})^p\}$}
\State{\Return $\widehat{f_j^{\ell}}$}
\end{algorithmic}
\end{algorithm}
We first show subsampling allows us to find a substream that satisfies the required assumptions. 
We can then boost the probability of success for estimating the frequency of each heavy hitter using a standard median-of-means argument. 
\begin{restatable}{lemma}{lemhighsamplehold}
\lemlab{lem:high:sample:hold}
Let $j\in[n]$ be an item with $(f_j)^p\ge\frac{\eps^2\cdot F_p}{2^{10}\gamma\log^2(nm)}$. 
Then with probability $1-\frac{1}{\poly(n)}$, $\FullSampleAndHold$ outputs $\widehat{f_j}$ such that
\[\left(1-\frac{\eps}{8\log(nm)}\right)\cdot(f_j)^p\le(\widetilde{f_j})^p\le\left(1+\frac{\eps}{8\log(nm)}\right)\cdot(f_j)^p.\]
\end{restatable}
\begin{proof}
Consider a fixed $j\in[n]$ with $(f_j)^p\ge\frac{\eps^2\cdot F_p}{2^{10}\gamma\log^2(nm)}$. 
Let $q>0$ be the integer such that $\frac{f_j}{2^q}\in\left[\frac{400}{\eps^2},\frac{800}{\eps^2}\right]$. 
Let $A_j$ be the random variable denoting the number of occurrences of $j$ in $J_q$ and note that over the randomness of the sampling, we have $\Ex{A_j}=\frac{f_j}{2^q}\le\frac{400}{\eps^2}$. 
We also have $\Ex{A_j^2}\le\frac{f_j}{2^q}+\frac{(f_j)^2}{2^{2q}}$, so that $\Var{A_j}\le\frac{400}{\eps^2}$. 
By Chebyshev's inequality, we have that the number of occurrences of $j$ in $J_q^{(r)}$ is a $(1+\eps)$-approximation of $\frac{f_j}{2^q}$ with probability at least $0.9$. 
We similarly have that with probability at least $0.9$, $m_x\in\left[\frac{(1-\eps)m}{2^q},\frac{(1+\eps)m}{2^q}\right]$ and thus $(A_j)^p\le\frac{800^p}{\eps^{2p}}\cdot m_q$. 

Since $(f_j)^p\ge\frac{\eps^2\cdot F_p}{2^{10}\gamma\log^2(nm)}$, then by a Chernoff bound, we have that for any $y\in[n]$, the number of occurrences of $y$ in $J_q^{(r)}$ is at most $\frac{\xi\log^3(nm)}{\eps^4}$ for some constant $\xi>1$, with high probability. 
Thus by a union bound, we have that with high probability, 
\[F_p(J_q^{(r)})\le\frac{\xi'n\log^{3p}(nm)}{\eps^{4p}},\]
for some sufficiently large constant $\xi'>1$, which satisfies the assumptions of \lemref{lem:one:sample:hold}. 
Hence we have with probability at least $0.99$, $\widehat{f_j^{(r,q)}}$ is a $(1+\eps)$-approximation to the number of occurrences of $j$ in $J_q^{(r)}$. 
Thus, $\widehat{f_j^{(r,q)}}\le\frac{2\cdot 800^p}{\eps^{2p}}\cdot m_q$ with probability at least $0.7$. 
Observe that since $(f_j)^p\ge\frac{\eps^2\cdot F_p}{2^{10}\gamma\log^2(nm)}$, then for any $y\in[n]$, we have that the expected number of occurrences of $y$ in $J_q^{(r)}$ is at most $\frac{\xi\log^2(nm)}{\eps^4}$, for some sufficiently large constant $\xi>1$. 
Hence by standard Chernoff bounds, we have that the median satisfies \[\widehat{f_j^{(q)}}=\median_{r\in[R]}\widehat{f_j^{(r,q)}}\le\frac{2\cdot 800^p}{\eps^{2p}}\cdot m_q,\]
with high probability. 

Moreover, observe that (1) for any stream with subsampling rate $\frac{1}{2^x}>\frac{1}{2^q}$, we similarly have that the number of occurrences of $j$ in $J_x$ is a $(1+\eps)$-approximation of $\frac{f_j}{2^x}$ with high probability and (2) $\SampleAndHold$ cannot overestimate the frequency of $j$. 
Thus, $\widehat{f_j^{(\ell)}}$ output by $\FullSampleAndHold$ satisfies \[\left(1-\frac{\eps}{8\log(nm)}\right)\cdot(f_j)^p\le(\widetilde{f_j})^p\le\left(1+\frac{\eps}{8\log(nm)}\right)\cdot(f_j)^p\]
with high probability. 
\end{proof}

Putting things together, we obtain \thmref{thm:hh:main}.
\thmhhmain*
\begin{proof}
Correctness of the algorithm for the heavy-hitters follows from \lemref{lem:high:sample:hold}. 
For the items that are not heavy-hitters, observe that by a standard concentration argument, it follows that with high probability across the independent instances, the median of the number of sampled items after rescaling will not surpass the threshold. 

We next analyze the space complexity. 
For $p\in[1,2]$, only $\cO{\frac{\log^{11+3p}(mn)}{\eps^{4+4p}}}$ counters are stored, while for $p>2$, only $\tO{\frac{1}{\eps^{4+4p}} n^{1-2/p}}$ counters are stored. 
Hence, the space complexity follows. 

It remains to analyze the number of internal state changes. 
The internal state can change each time an item is sampled. 
Since each item of the stream is sampled with probability $\varrho=\frac{\gamma^2n^{1-1/p}\log^4(nm)}{\eps^2 m}$, then with high probability, the total number of internal state changes is $\frac{\gamma^2n^{1-1/p}\log^4(nm)}{\eps^2}$. 
\end{proof}

\section{\texorpdfstring{$F_p$ Estimation}{Fp Estimation}}
\seclab{sec:fp}
In this section, we present insertion-only streaming algorithms for $F_p$ estimation with a small number of internal state changes. 
We first observe that an $F_p$ estimation algorithm by \cite{JayaramW19} achieves a small number of internal state changes for $p<1$. 
We then build upon our $L_p$-heavy hitter algorithm to achieve an $F_p$ estimation algorithm that achieves a small number of internal state changes for $p>1$.

\subsection{\texorpdfstring{$F_p$ Estimation, $p<1$}{Fp Estimation, p<1}}
As a warm-up, we first show that the $F_p$ estimation streaming algorithm of \cite{JayaramW19} uses a small number of internal state changes for $p<1$. 
We first recall the following definition of the $p$-stable distribution:
\begin{definition}[$p$-stable distribution]
\cite{Zolotarev89}
\deflab{def:pstable}
For $0<p\le 2$, the $p$-stable distribution $\calD_p$ exists and satisfies $\sum_{i=1}^n Z_ix_i\sim\|x\|_p\cdot Z$ for $Z,Z_1,\ldots,Z_n\sim\calD_p$ and any vector $x\in\mathbb{R}^n$. 
\end{definition}
A standard method~\cite{Nolan03} for generating $p$-stable random variables is to first generate $\theta\sim\Uni\left(\left[-\frac{\pi}{2},\frac{\pi}{2}\right]\right)$ and $r\sim\Uni([0,1])$ and then set
\[X=f(r,\theta)=\frac{\sin(p\theta)}{\cos^{1/p}(\theta)}\cdot\left(\frac{\cos(\theta(1-p))}{\log\frac{1}{r}}\right)^{\frac{1}{p}-1}.\]

The $F_p$ estimation streaming algorithm of \cite{JayaramW19} first generates a sketch matrix $D\in\mathbb{R}^{k\times n}$, where $k=\cO{\frac{1}{\eps^2}}$ and each entry of $D$ is generated from the $p$-stable distribution. 
Observe that $D$ can be viewed as $k$ vectors $D^{(1)},\ldots,D^{(k)}\in\mathbb{R}^n$ of $p$-stable random variables. 
For $i\in[k]$, suppose we maintained $\langle D^{(i)}, x\rangle$, where $x$ is the frequency vector induced by the stream. 
Then it is known \cite{Indyk06} that with constant probability, the median of these inner products is a $(1+\eps)$-approximation to $F_p$. 

\cite{JayaramW19} notes that each vector $D^{(i)}$ can be further decomposed into a vector $D^{(i,+)}$ containing the positive entries of $D^{(i)}$ and a vector $D^{(i,-)}$ containing the negative entries of $D^{(i)}$. 
Since $D^{(i)}=D^{(i,+)}+D^{(i,-)}$, then it suffices to maintain $\langle D^{(i,+)}, x\rangle$ and $\langle D^{(i,-)}, x\rangle$ for each $i\in[k]$. 
For insertion-only streams, all entries of $x$ are non-negative, and so the inner products $\langle D^{(i,+)}, x\rangle$ and $\langle D^{(i,-)}, x\rangle$ are both monotonic over the course of the stream, which permits the application of Morris counters. 
Thus the algorithm of \cite{JayaramW19} instead uses Morris counters to approximately compute $\langle D^{(i,+)}, x\rangle$ and $\langle D^{(i,-)}, x\rangle$ to within a $(1+\cO{\eps})$-multiplicative factor. 
The key technical point is that \cite{JayaramW19} shows that
\[\left\lvert\langle D^{(i,-)}, x\rangle\right\rvert+\left\lvert\langle D^{(i,-)}, x\rangle\right\rvert=\cO{\|x\|_p}\]
for $p<1$ and so $(1+\cO{\eps})$-multiplicative factor approximations to $\langle D^{(i,+)}, x\rangle$ and $\langle D^{(i,-)}, x\rangle$ are enough to achieve a $(1+\eps)$-approximation to $\langle D^{(i)}, x\rangle$. 
Now the main gain is that using Morris counters to approximate $\langle D^{(i,+)}, x\rangle$ and $\langle D^{(i,-)}, x\rangle$, not only is the overall space usage improved for the purposes of \cite{JayaramW19}, but also for our purposes, the number of internal state updates is much smaller. 

As an additional technical caveat, \cite{JayaramW19} notes that the sketching matrix $D$ cannot be stored in the allotted memory. 
Instead, \cite{JayaramW19} notes that by using the log-cosine estimator~\cite{KaneNW10} instead of the median estimator, the entries of $D$ can be generated using $\cO{\frac{\log(1/\eps)}{\log\log(1/\eps)}}$-wise independence, so that storing the randomness used to generate $D$ only requires $\cO{\frac{\log(1/\eps)}{\log\log(1/\eps)}\log n}$ bits of space. 

\begin{theorem}
\thmlab{thm:small:p}
For $p\in(0,1]$, there exists a one-pass insertion-only streaming algorithm that uses $\poly\left(\log n,\frac{1}{\eps},\log\frac{1}{\delta}\right)$ internal state changes, $\cO{\frac{1}{\eps^2}\left(\log\log n+\log\frac{1}{\eps}\right)+\frac{\log(1/\eps)}{\log\log(1/\eps)}\log n}$ bits of space, and outputs a $(1+\eps)$-approximation to $F_p$ with high probability. 
\end{theorem}

\subsection{\texorpdfstring{$F_p$ Estimation, $p>1$}{Fp Estimation, p>1}}
In this section, we present our $F_p$ approximation algorithm for insertion-only streams that only have $\widetilde{\Omega}_\eps(n^{1-1/p})$ internal state updates for $p>1$.  

We first define the level sets of the $F_p$ moment, as well as the contribution of each level set. 
\begin{definition}[Level sets and contribution]
Let $\widetilde{F_p}$ be the power of two such that $F_p\le\widetilde{F_p}<2F_p$. 
Given a uniformly random $\lambda\in\left[\frac{1}{2},1\right]$, for each $\ell\in[L]$, we define the level set 
\[\Gamma_\ell:=\left\{i\in[n]\,\mid\,f_i\in\left[\frac{\lambda\cdot\widetilde{F_p}}{2^\ell},\frac{2\lambda\widetilde{F_p}}{2^\ell}\right)\right\}.\]
We define the contribution $C_\ell$ and the fractional contribution $\phi_\ell$ of level set $\Gamma_\ell$ to be $C_\ell:=\sum_{i\in\Gamma_\ell}(f_i)^p$ and $\phi_\ell:=\frac{C_\ell}{F_p}$. 

For an accuracy parameter $\eps$ and a stream of length $m$, we say that a level set $\Gamma_\ell$ is significant if its fractional contribution $\phi_\ell$ is at least $\frac{\eps}{2p\log(nm)}$. 
Otherwise, we say the level set is insignificant.
\end{definition}

Our algorithm follows from the framework introduced by \cite{IndykW05} and subsequently used in a number of different applications, e.g.,~\cite{WoodruffZ12,BlasiokBCKY17,LevinSW18,BravermanWZ21,WoodruffZ21,MahabadiWZ22,BravermanMWZ23} and has the following intuition. 
We estimate the contributions of the significant level sets by approximating the frequencies of the heavy hitters for substreams induced by subsampling the universe at exponentially smaller rates. 
Specifically, we create $L=\cO{\log n}$ substreams where for each $\ell\in[L]$, we subsample each element of the universe $[n]$ into the substream with probability $\frac{1}{2^{\ell-1}}$. 
We rescale $(1+\eps)$-approximations to the contributions of the surviving heavy hitters by the inverse of the sampling rate to obtain good approximations of the contributions of each significant level set. 

To guarantee a small number of internal state changes, we use our heavy hitter algorithm $\FullSampleAndHold$ to provide $(1+\eps)$-approximations to the heavy hitters in each substream, thereby obtaining good approximations to the contributions of each significant level set.  
Our algorithm appears in full in \algref{alg:fp:est}. 

\begin{algorithm}[!htb]
\caption{$F_p$ approximation algorithm, $p\ge 1$}
\alglab{alg:fp:est}
\begin{algorithmic}[1]
\Require{Stream $s_1,\ldots,s_m$ of items from $[n]$, accuracy parameter $\eps\in(0,1)$, $p\ge 1$}
\Ensure{$(1+\eps)$-approximation to $F_p$}
\State{$L=\cO{p\log(nm)}$, $R=\cO{\log\log n}$, $\gamma=2^{20}p$}
\For{$t=1$ to $t=m$}
\For{$(\ell,r)\in[L]\times[R]$}
\State{Let $m_\ell^{(r)}$ be a $2$-approximation to the length of the induced stream}
\Comment{Morris counter}
\State{Let $I_\ell^{(r)}$ be a (nested) subset of $[n]$ subsampled at rate $p_\ell:=\min(1, 2^{1-\ell})$}
\If{$s_t\in I_\ell^{(r)}$}
\State{Send $s_t$ to $\FullSampleAndHold_\ell^{(r)}$}
\EndIf
\EndFor
\EndFor
\State{Let $H_i^{(r)}$ be the outputs of the Morris counters at level $i$}
\State{Let $\widetilde{M}$ be the power of two such that $m^p\le\widetilde{M}<2m^p$}
\State{Let $S_i^{(r)}$ be the set of ordered pairs $(j,\widehat{f_j})$ of $H_i^{(r)}$ with $\left(\widehat{f_j}\right)^p\in\left[\frac{\lambda\cdot\widetilde{M}}{2^i},\frac{2\lambda\cdot\widetilde{M}}{2^i}\right]$}
\For{$i=1$ to $i=L$}
\State{$\ell\gets\max\left(1,i-\flr{\log\frac{\gamma^2\log(nm)}{\eps^2}}\right)$}
\State{$\widehat{C_i}\gets\frac{1}{p_\ell}\median_{r\in[R]}\left(\sum_{(j,\widehat{f_j})\in S_\ell^{(r)}}\left(\widehat{f_j}\right)^p\right)$}
\EndFor
\State{\Return $\widehat{F_p}=\sum_{\ell\in[L]}\widehat{C_\ell}$}
\end{algorithmic}
\end{algorithm}

We note the following corollary of \lemref{lem:high:sample:hold}. 
\begin{restatable}{lemma}{lemhh}
\lemlab{lem:hh}
Let $r\in[R]$ be fixed. 
Suppose $j\in I_\ell^{(r)}$ and $(f_j)^p\ge\frac{\eps^2\cdot F_p(I_\ell^{(r)})}{2^7\gamma\log^2(nm)}$. 
Then with probability at least $\frac{9}{10}$, $H_\ell^{(r)}$ outputs $\widehat{f_j}$ with
\[\left(1-\frac{\eps}{8\log(nm)}\right)\cdot(f_j)^p\le(\widetilde{f_j})^p\le\left(1+\frac{\eps}{8\log(nm)}\right)\cdot(f_j)^p.\]
\end{restatable}
\begin{proof}
The proof follows from \lemref{lem:high:sample:hold} and the fact that $\FullSampleAndHold$ is only run on the substream induced by $I_\ell^{(r)}$. 
\end{proof}

\begin{restatable}{lemma}{lemfreqacc}
\lemlab{lem:freq:acc}
Let $\eps\in(0,1)$, $\Gamma_i$ be a fixed level set and let $\ell:=\max\left(1,i-\flr{\log\frac{\gamma\log(nm)}{\eps^2}}\right)$. 
For a fixed $r\in[R]$, let $\calE_1$ be the event that $|I_\ell^{(r)}|\le\frac{32n}{2^\ell}$ and let $\calE_2$ be the event that $F_p(I_\ell^{(r)})\le\frac{32F_p}{2^\ell}$. 
Then conditioned on $\calE_1$ and $\calE_2$, for each $j\in\Gamma_i\cap I_\ell^{(r)}$, there exists $(j,\widetilde{f_j})$ in $S_i^{(r)}$ such that with probability at least $\frac{9}{10}$,
\[\left(1-\frac{\eps}{8\log(nm)}\right)\cdot(f_j)^p\le(\widetilde{f_j})^p\le\left(1+\frac{\eps}{8\log(nm)}\right)\cdot(f_j)^p.\]
\end{restatable}
\begin{proof}
We consider casework on whether $i-\flr{\log\frac{\gamma\log^2(nm)}{\eps^2}}\le 1$ or $i-\flr{\log\frac{\gamma\log^2(nm)}{\eps^2}}>1$. 
This corresponds to whether the frequencies $\left(\widehat{f_j}\right)^p$ in a significant level set are large or not large, informally speaking. 
If the frequencies are large, then it suffices to estimate them using our sampling-based algorithm. 
However, if the frequencies are not large, then subsampling must first be performed before we can estimate the frequencies using our sampling-based algorithm.

Suppose $i-\flr{\log\frac{\gamma\log^2(nm)}{\eps^2}}\le 1$, so that $\frac{1}{2^i}\ge\frac{\eps^2}{\gamma\log^2(nm)}$. 
Since $j\in\Gamma_i$, we have $(f_j)^p\in\left[\frac{\lambda\cdot\widetilde{F_p}}{2^i},\frac{2\lambda\widetilde{F_p}}{2^i}\right]$ and thus
\[(f_j)^p\ge\frac{\eps^2\cdot\widetilde{F_p}}{\gamma\log^2(nm)},\qquad f_j\ge\frac{\eps^{2/p}\cdot\left(\widetilde{F_p}\right)^{2/p}}{\gamma^{1/p}\log^{1/p}(nm)}\]
Moreover, for $\ell=\max\left(1,i-\flr{\log\frac{\gamma\log^2(nm)}{\eps^2}}\right)$, we have $\ell=1$, so we consider the outputs by the Morris counters $H_1^{(r)}$. 
By \lemref{lem:hh}, we have that with probability at least $\frac{9}{10}$, $H_\ell^{(r)}$ outputs $\widehat{f_j}$ such that
\[\left(1-\frac{\eps}{8\log(nm)}\right)\cdot(f_j)^p\le(\widetilde{f_j})^p\le\left(1+\frac{\eps}{8\log(nm)}\right)\cdot(f_j)^p,\]
as desired. 

For the other case, suppose $i-\flr{\log\frac{\gamma\log^2(nm)}{\eps^2}}>1$, so that $\ell=i-\flr{\log\frac{\gamma\log^2(nm)}{\eps^2}}$ and $p_\ell=2^{1-\ell}$. 
Therefore, 
\[\frac{1}{2^\ell}=\frac{\gamma\log^2(nm)}{\eps^2}\frac{1}{2^i}.\]
Since $j\in\Gamma_i$, we have again $(f_j)^p\in\left[\frac{\lambda\cdot\widetilde{F_p}}{2^i},\frac{2\lambda\widetilde{F_p}}{2^i}\right]$ and therefore,
\[(f_j)^p\ge\frac{F_p}{4\cdot2^i}\ge\frac{\eps^2}{4\gamma\log^2(nm)}\frac{F_p}{2^\ell}.\]
Conditioning on the event $\calE_2$, we have $F_p(I_\ell^{(r)})\le\frac{32F_p}{2^\ell}$ and thus
\[(f_j)^p\ge\frac{F_p}{4\cdot2^i}\ge\frac{\eps^2}{128\gamma\log^2 n}\cdot F_p(I_\ell^{(r)}).\]
Therefore by \lemref{lem:hh}, we have that with probability at least $\frac{9}{10}$, $H_\ell^{(r)}$ outputs $\widehat{f_j}$ such that
\[\left(1-\frac{\eps}{8\log(nm)}\right)\cdot(f_j)^p\le(\widetilde{f_j})^p\le\left(1+\frac{\eps}{8\log(nm)}\right)\cdot(f_j)^p,\]
as desired. 
\end{proof}

We now justify the approximation guarantees of our algorithm. 
\begin{restatable}{lemma}{lemcorrectness}
\lemlab{lem:correctness}
$\PPr{\left\lvert\widehat{F_p}-F_p\right\rvert\le\eps\cdot F_p}\ge\frac{2}{3}$.
\end{restatable}
\begin{proof}
We would like to show that for each level set $i$, we accurately estimate its contribution $C_i$, i.e., we would like to show $\vert\widehat{C_i}-C_i\rvert\le\frac{\eps}{8\log(nm)}\cdot F_p$ for all $i$. 
For a fixed $i$, recall that $C_i=\sum_{j\in\Gamma i} (f_j)^p$, where $j\in\Gamma_i$ if $(f_j)^p\in\left[\frac{\lambda\cdot\widetilde{F_p}}{2^i},\frac{2\lambda\widetilde{F_p}}{2^i}\right)$. 
On the other hand, $\widehat{C_i}$ is a scaled sum of items $j$ whose estimated frequency satisfies $\left(\widehat{f_j}\right)^p\in\left[\frac{\lambda\cdot\widetilde{M}}{2^i},\frac{2\lambda\cdot\widetilde{M}}{2^i}\right]$. 
Then $j$ could be classified into contributing to $\widehat{C_i}$ even if $j\notin\Gamma_i$. 
Thus we first consider an idealized process where $j$ is correctly classified across all level sets and show that in this idealized process, we achieve a $(1+\cO{\eps})$-approximation to $F_p$. 
We then argue that because we choose $\lambda$ uniformly at random, then only a small number of coordinates will be misclassified and so our approximation guarantee will only slightly degrade, but remain a $(1+\eps)$-approximation to $F_p$. 

\textbf{Idealized process.}
We first show that in a setting where $(\widehat{f_j})^p$ is correctly classified for all $j$, then for a fixed level set $i$, we have $\vert\widehat{C_i}-C_i\rvert\le\frac{\eps}{8\log(nm)}\cdot F_p$ with probability $1-\frac{1}{\poly(nm)}$. 

For a fixed $r\in[R]$, let $\calE_1$ be the event that $|I_i^{(r)}|\le\frac{32n}{2^i}$ and let $\calE_2$ be the event that $F_p(I_i^{(r)})\le\frac{32F_p}{2^i}$. 
Let $\calE_3$ be the event that 
\[\left(1-\frac{\eps}{8\log(nm)}\right)\cdot(f_j)^p\le(\widetilde{f_j})^p\le\left(1+\frac{\eps}{8\log(nm)}\right)\cdot(f_j)^p.\]
Conditioned on $\calE_1$, $\calE_2$, and $j\in I_i^{(r)}$, then we have that $\PPr{\calE_3}\ge\frac{9}{10}$, by \lemref{lem:freq:acc}. 

We define $\widehat{C_i^{(r)}}:=\frac{1}{p_\ell}\sum_{(j,\widehat{f_j})}\in S_i^{(r)}\left(\widehat{f_j}\right)^p$. 
Therefore, we have that $\widehat{C_i}=\median_{r\in[R]}\widehat{C_i^{(r)}}$, recalling $\ell=\max\left(1,i-\flr{\log\frac{\gamma\log(nm)}{\eps^2}}\right)$. 

Conditioned on $\calE_3$, we have
\[\left(1-\frac{\eps}{8\log(nm)}\right)D_i^{(r)}\le\widehat{C_i^{(r)}}\le\left(1+\frac{\eps}{8\log(nm)}\right)D_i^{(r)},\]
where we define
\[D_i^{(r)}=\frac{1}{p_\ell}\sum_{j\in S_i^{(r)}\cap \Gamma_i}(f_j)^p.\]
Note that
\[\Ex{D_i^{(r)}}=\frac{1}{p_\ell}\sum_{j\in\Gamma_i}p_\ell\cdot(f_j)^p=C_i,\]
where we recall that $C_i$ denotes the contribution of level set $\Gamma_i$. 
We also have
\[\Var{D_i^{(r)}}\le\frac{1}{(p_\ell)^2}\sum_{j\in\Gamma_i}p_\ell\cdot(f_j)^{2p}\le\frac{\eps^2}{2^{i-1}\cdot(\gamma^2\log(nm))}\sum_{j\in C_i}(f_j)^{2p}.\]
For $j\in\Gamma_i$, we have $(f_j)^{2p}\le\frac{4\lambda^2(\widetilde{F_p})^2}{2^{2i}}\le\frac{16(F_p)^2}{2^{2i}}$, since $\lambda\le 1$ and $\widetilde{F_p}\le 2F_p$. 
Therefore for sufficiently large $\gamma$, 
\[\Var{D_i^{(r)}}\le\frac{\eps^2|\Gamma_i|}{100p^2\cdot 2^i\cdot\log^2(nm)}(F_p)^2.\]
Since $\frac{|\Gamma_i|}{2^i}\le\phi_i\le 1$, then
\[\Var{D_i^{(r)}}\le\frac{\eps^2}{10000p\log(nm)}\cdot(F_p)^2.\]
Thus by Chebyshev's inequality, we have
\[\PPr{\left\lvert D_i^{(r)}-C_i\right\rvert\ge\frac{\eps}{10p\log(nm)}\cdot F_p}\le\frac{1}{10}.\]
Therefore, 
\[\PPr{\left\lvert\widehat{C_i^{(r)}}-C_i\right\rvert\le\frac{\eps}{10p\log(nm)}\cdot F_p\,\mid\,\calE_1\wedge\calE_2}\ge\frac{4}{5}.\]

To analyze the probability of the events $\calE_1$ and $\calE_2$ occurring, note that in level $\ell$, each item is sampled with probability $2^{-\ell+1}$. 
Hence, 
\[\Ex{|I_\ell^{(r)}|}\le\frac{n}{2^{\ell-1}},\qquad\Ex{F_p(I_\ell^{(r)})}\le\frac{F_p}{2^{\ell-1}}.\]
Since $\calE_1$ is the event that $|I_\ell^{(r)}|\le\frac{32n}{2^\ell}$ and $\calE_2$ is the event that $F_p(I_\ell^{(r)})\le\frac{32F_p}{2^\ell}$, then by Markov's inequality, we have
\[\PPr{E_1}\ge\frac{15}{16},\qquad\PPr{E_2}\ge\frac{15}{16}.\] 
By a union bound,
\[\PPr{\left\lvert\widehat{C_i^{(r)}}-C_i\right\rvert\le\frac{\eps}{10p\log(nm)}\cdot F_p}\ge 0.676.\]

Since $\widehat{C_i}=\median_{r\in[R]}\widehat{C_i^{(r)}}$ over $R=\cO{\log\log n}$ independent instances, then we have
\[\PPr{\left\lvert\widehat{C_i}-C_i\right\rvert\le\frac{\eps}{10p\log(nm)}\cdot F_p}\ge1-\frac{1}{\polylog(n)}.\]
Hence by a union bound over the $p\log(nm)$ level sets,
\begin{align*}
\left\lvert\widehat{F_p}-F_p\right\rvert&=\left\lvert\sum_{i=1}^{p\log(nm)}\widehat{C_i}-\sum_{i=1}^{p\log(nm)} C_i|\right\rvert\le\sum_{i=1}^{p\log(nm)}\left\lvert\widehat{C_i}-C_i|\right\rvert\\
&\le\sum_{i=1}^{p\log(nm)}\frac{\eps}{10p\log(nm)}\cdot F_p\le\frac{\eps}{10}\cdot F_p.
\end{align*}

\textbf{Randomized boundaries.}
Given a fixed $r\in[R]$, we say that an item $j\in[n]$ is misclassified if there exists a level set $\Gamma_i$ such that 
\[(f_j)^p\in\left[\frac{\lambda\cdot\widetilde{F_p}}{2^i},\frac{2\lambda\cdot\widetilde{F_p}}{2^i}\right],\]
but for the estimate $\widetilde{f_j}$, we have either \[(\widetilde{f_j})^p<\frac{\lambda\cdot\widetilde{F_p}}{2^i}\qquad\text{or}\qquad(\widetilde{f_j})^p\ge\frac{2\lambda\cdot\widetilde{F_p}}{2^i}.\] 
By \lemref{lem:hh}, we have that conditioned on $\calE_3$,
\[\left(1-\frac{\eps}{8\log(nm)}\right)\cdot(f_j)^p\le(\widetilde{f_j})^p\le\left(1+\frac{\eps}{8\log(nm)}\right)\cdot(f_j)^p,\]
independently of the choice of $\lambda$. 
Since $\lambda\in\left[\frac{1}{2},1\right]$ is chosen uniformly at random, then the probability that $j\in[n]$ is misclassified is at most $\frac{\eps}{2\log(nm)}$. 

Furthermore, if $j\in[n]$ is misclassified, then it can only be classified into either level set $\Gamma_{i+1}$ or level set $\Gamma_{i-1}$, because $\left(\widehat{f_j}^p\right)$ is a $\left(1+\frac{\eps}{8\log(nm)}\right)$-approximation to $(f_j)^p$. 
Thus, a misclassified index induces at most $2(f_j)^p$ additive error to the contribution of level set $\Gamma_i$. 
In expectation across all $j\in[n]$, the total additive error due to misclassification is at most $F_p\cdot\frac{\eps}{2\log(nm)}$. 
Therefore by Markov's inequality for sufficiently large $n$ and $m$, the total additive error due to misclassification is at most $\frac{\eps}{2}\cdot F_p$ with probability at least $0.999$. 
Hence in total, 
\[\PPr{\left\lvert\widehat{F_p}-F_p\right\rvert\le\eps\cdot F_p}\ge\frac{2}{3}.\]
\end{proof}

Putting things together, we give the full guarantees of our $F_p$ estimation algorithm in \thmref{thm:fp:main}. 

\thmfpmain*
\begin{proof}
The space bound follows from the fact that for $p\in[1,2]$, only $\cO{\frac{\log^{11+3p}(mn)}{\eps^{4+4p}}}$ counters are stored, while for $p>2$, only $\tO{\frac{1}{\eps^{4+4p}} n^{1-2/p}}$ counters are stored. 
The internal state can change each time an item is sampled. 
Since each item of the stream is sampled with probability $\varrho=\frac{\gamma^2n^{1-1/p}\log^4(nm)}{\eps^2 m}$, then with high probability, the total number of internal state changes is $\frac{\gamma^2n^{1-1/p}\log^4(nm)}{\eps^2}$. 

Finally for correctness, we have by \lemref{lem:correctness}, that 
\[\PPr{\left\lvert\widehat{F_p}-F_p\right\rvert\le\eps\cdot F_p}\ge\frac{2}{3}.\]
\end{proof}

\subsection{Entropy Estimation}
In this section, we describe how to estimate the entropy of a stream using a small number of internal state changes. 
Recall that for a frequency vector $f\in\mathbb{R}^n$, the Shannon entropy of $f$ is defined by $H(f)=-\sum_{i=1}^n f_i\log f_i$. 
Observe that any algorithm that obtains a $(1+\cO{\eps})$-multiplicative approximation to the function $h(f)=2^{H(f)}$ also obtains an $\cO{\eps}$-additive approximation of the Shannon entropy $H(f)$, and vice versa. 
Hence to obtain an additive $\eps$-approximation to the Shannon entropy, we describe how to obtain a multiplicative $(1+\eps)$-approximation to $h(f)=2^{H(f)}$. 
\begin{lemma}[\cite{HarveyNO08}]
\lemlab{lem:entropy:reduction}
Given an accuracy parameter $\eps>0$, let $k=\log\frac{1}{\eps}+\log\log m$ and $\eps'=\frac{\eps}{12(k+1)^3\log m}$. 
Then there exists an efficiently computable set $\{p_0,\ldots,p_k\}$ such that $p_i\in(0,2)$ for all $i$, as well as an efficiently computable deterministic function that uses $(1+\eps')$-approximations to $F_{p_i}(f)$ to compute a $(1+\eps)$-approximation to $h(f)=2^{H(f)}$. 
\end{lemma}
Section 3.3 in \cite{HarveyNO08} describes how to compute the set $\{p_0,\ldots,p_k\}$ in \lemref{lem:entropy:reduction} as follows. 
We define $\ell=\frac{1}{2(k+1)\log m}$ and the function $g(z)=\frac{\ell(k^2(z-1)+1))}{2k^2+1}$. 
For each $p_i$, we set $p_i=1+g(\cos(i\pi/k))$, which can be efficiently computed. 
Thus, the set $\{p_0,\ldots,p_k\}$ in \lemref{lem:entropy:reduction} can be efficiently computed as pre-processing, assuming that $n$ and $m$ are known a priori. 
Let $P(x)$ be the degree $k$ polynomial interpolated at the points $p_0,\ldots,p_k$, so that $P(p_i)=F_{p_i}(f)$ for each $i\in[k]$, where $F_{p_i}(f)$ is the $(p_i)$-th moment of the frequency vector $f$. 
\cite{HarveyNO08} then showed that a multiplicative $(1+\cO{\eps})$-approximation to $h(f)=2^{H(f)}$ can then be computed from $2^{P(0)}$, and moreover a $(1+\cO{\eps})$-approximation to $2^{P(0)}$ can be computed from $(1+\cO{\eps})$-approximations to $F_{p_i}(f)$, for each $i\in[k]$. 

Thus by \lemref{lem:entropy:reduction} and \thmref{thm:fp:main}, we have:
\begin{theorem}
Given an accuracy parameter $\eps\in(0,1)$, as well as the universe size $n$, and the stream length $m=\poly(n)$, there exists a one-pass insertion-only streaming algorithm that has $\tO{\frac{1}{\eps^{\cO{1}}}\sqrt{n}}$ internal state changes, uses $\cO{\frac{\log^{\cO{1}}(mn)}{\eps^{\cO{1}}}}$ bits of space, and outputs $\widehat{H}$ such that $\PPr{\left\lvert\widehat{H}-H\right\rvert\le\eps}\ge\frac{2}{3}$, where $H$ is the Shannon entropy of the stream. 
\end{theorem}

\section{Lower Bound}
\seclab{sec:lab}
In this section, we describe our lower bound showing that any streaming algorithm achieving a $(2-\Omega(1))$-approximation to $F_p$ requires at least $\frac{1}{2}n^{1-1/p}$ state updates, regardless of the memory allocated to the algorithm. 
The proof of \thmref{thm:hh:lb} is similar. 
The main idea is that we create two streams $\calS_1$ and $\calS_2$ of length $\cO{n}$ that look similar everywhere except for a random contiguous block $B$ of $n^{1/p}$. 
In $B$, the first stream $\calS_1$ has the same item repeated $n^{1/p}$ times, while the second stream $\calS_2$ has $n^{1/p}$ distinct items each appear once. 
The remaining $n-n^{1/p}$ stream updates of $\calS_1$ and $\calS_2$ are additional distinct items that each appear once, so that $F_p(\calS_1)\ge(2-o(1))\cdot F_p(\calS_2)$ and $F_p(\calS_2)=\Omega(n)$. 
Any algorithm $\calA$ that achieves a $(2-\Omega(1))$-approximation to $F_p$ must distinguish between $\calS_1$ and $\calS_2$ and thus $\calA$ must perform some action on $B$. 
However, $B$ has size $n^{1/p}$ and has random location throughout the stream, so $\calA$ must perform $\Omega(n^{1-1/p})$ state updates. 
\thmfpconslb*
\begin{proof}
Consider the two following possible input streams. 
For the stream $\calS_1$ of length $n$ on universe $n$, we choose a random contiguous block $B$ of $n^{1/p}$ stream updates and set them all equal to the same random universe item $i\in[n]$. 
We randomly choose the remaining $n-n^{1/p}$ updates in the stream so that they are all distinct and none of them are equal to $i$. 
Note that by construction, we have $F_p(\calS_1)=(n-n^{1/p})+(n^{1/p})^p=2n-n^{1/p}$. 
For the stream $\calS_2$ of length $n$ on universe $n$, we choose $\calS_2$ to be a random permutation of $[n]$, so that $F_p(\calS_2)=n$.

For fixed $\eps\in(0,1)$, let $\calA$ be an algorithm that achieves a $2-\eps$ approximation to $F_p$ with probability at least $\frac{2}{3}$, while using fewer than $\frac{1}{2}n^{1-1/p}$ state updates. 
We claim that with probability $\frac{1}{2}$, $\calA$ must have the same internal state before and after $B$ in the stream $\calS_1$. 
Note that we can view each stream update as $(i,j)$ where $i\in[n]$ is the index of the stream update and $j\in[n]$ is the identity of the universe item. 
Observe that for a random stream update $i\in[n]$, a random universe update $j\in[n]$ alters the state of $\calA$ with probability at most $\frac{1}{2n^{1/p}}$, since otherwise for a random stream, the expected number of state changes would be larger than $\frac{n^{1-1/p}}{2}$, which would be a contradiction. 
Since both the choice of $B$ and the element $j\in[n]$ that is repeated $n^{1/p}$ times are chosen uniformly at random, then the expected number of changes of the streaming algorithm on the block $B$ is at most $\frac{n^{1/p}}{2 n^{1/p}}=\frac{1}{2}$. 
Therefore, with probability at least $\frac{1}{2}$, the streaming algorithm's state is the same before and after the block $B$. 

Moreover, the same argument applies to $\calS_2$, and so therefore with probability at least $\frac{1}{2}$, the streaming algorithm cannot distinguish between $\calS_1$ and $\calS_2$ if its internal state only changes fewer than $\frac{n^{1-1/p}}{2}$ times. 
\end{proof}
We now prove \thmref{thm:hh:lb} using a similar approach.
\thmhhlb*
\begin{proof}
Consider the two input streams $\calS_1$ and $\calS_2$ defined as follows.  
We define the stream $\calS_1$ to have length $n$ on a universe of size $n$. 
Similar to before, we choose a random contiguous block $B$ of $\eps\cdot n^{1/p}$ stream updates and set them all equal to the same random universe item $i\in[n]$. 
We randomly choose the remaining $n-n^{1/p}$ updates in the stream so that they are all distinct and none of them are equal to $i$. 
Note that by construction, we have 
\[F_p(\calS_1)=(n-n^{1/p})+(n^{1/p})^p=2n-n^{1/p}.\]
Therefore, the item replicated $\eps\cdot n^{1/p}$ times in block $B$ is an $\frac{\eps}{2}$-heavy hitter with respect to $L_p$. 

We also define the stream $\calS_2$ to have length $n$ on universe $n$. 
As before, we choose $\calS_2$ to be a random permutation of $[n]$, so that  $F_p(\calS_2)=n$.

For fixed $\eps\in(0,1)$, let $\calA$ be an algorithm that solves the $\eps$-heavy hitter problem with respect to $L_p$, with probability at least $\frac{2}{3}$, while using fewer than $\frac{1}{2\eps}n^{1-1/p}$ state updates. 
We claim that with probability $\frac{1}{2}$, $\calA$ must have the same internal state before and after $B$ in the stream $\calS_1$. 

Observe that each stream update can be viewed as $(i,j)$ where $i\in[n]$ is the index of the stream update and $j\in[n]$ is the identity of the universe item. 
For a random universe item $i\in[n]$, the probability that a random stream update $j\in[n]$ alters the state of $\calA$ is at most $\frac{1}{2\eps\cdot n^{1/p}}$, since otherwise for a random stream, the expected number of state changes would be larger than $\frac{n^{1-1/p}}{2\eps}$, which would be a contradiction. 
Because both the choice of $B$ is uniformly and the element $j\in[n]$ that is repeated $\eps\cdot n^{1/p}$ times are chosen uniformly at random, then the expected number of changes of the streaming algorithm on the block $B$ is at most $\frac{\eps\cdot n^{1/p}}{2\eps\cdot n^{1/p}}=\frac{1}{2}$. 
Hence, the streaming algorithm's state is the same before and after the block $B$, with probability at least $\frac{1}{2}$. 

However, the same argument applies to $\calS_2$. 
Thus with probability at least $\frac{1}{2}$, the streaming algorithm cannot distinguish between $\calS_1$ and $\calS_2$ if its internal state only changes fewer than $\frac{n^{1-1/p}}{2\eps}$ times. 
Therefore, any algorithm that solves the $L_p$-heavy hitter problem with threshold $\eps$ with probability at least $\frac{2}{3}$ requires at least $\frac{1}{2\eps}n^{1-1/p}$ state updates.
\end{proof}

\section*{Acknowledgements}
We would like to thank Mark Braverman for asking questions related to the ones we studied in this work. 
David P. Woodruff would like to thank Google Research and the Simons Institute for the Theory of Computing, where part of this work was done, as well as a Simons Investigator Award. 
David P. Woodruff and Samson Zhou are supported in part by NSF CCF-2335411. 
Part of this work was done while Samson Zhou was a postdoctoral researcher at UC Berkeley. 

\def\shortbib{0}
\bibliographystyle{alpha}
\bibliography{references}
\end{document}

%% file: preamble.tex
\pdfoutput=1
\usepackage[T1]{fontenc}
\usepackage{lmodern}
\usepackage[protrusion=true,expansion=true]{microtype}
\usepackage{amsmath,amssymb,amsfonts,amsthm}
\usepackage{subcaption}
\usepackage{graphicx}
\usepackage{setspace}
\usepackage{fullpage}
\usepackage[backref=page]{hyperref}
\usepackage{color}
\usepackage{wrapfig}
\usepackage{tikz}
\usepackage{quoting}
\usetikzlibrary{decorations.pathreplacing}
\usepackage{algorithm}
\usepackage[noend]{algpseudocode}
\usepackage[framemethod=tikz]{mdframed}
\usepackage{xspace}
\usepackage{pgfplots}
\usepackage{framed}
\usepackage{thmtools}
\usepackage{thm-restate}
\pgfplotsset{compat=1.5}

\newtheorem{theorem}{Theorem}[section]

\newtheorem{lemma}[theorem]{Lemma}

\newtheorem{definition}[theorem]{Definition}

\newenvironment{proofof}[1]{\begin{trivlist} \item {\bf Proof
#1:~~}}
  {\qed\end{trivlist}}

\newcommand{\namedref}[2]{\hyperref[#2]{#1~\ref*{#2}}}
\newcommand{\thmlab}[1]{\label{thm:#1}}
\newcommand{\thmref}[1]{\namedref{Theorem}{thm:#1}}
\newcommand{\lemlab}[1]{\label{lem:#1}}
\newcommand{\lemref}[1]{\namedref{Lemma}{lem:#1}}

\newcommand{\seclab}[1]{\label{sec:#1}}
\newcommand{\secref}[1]{\namedref{Section}{sec:#1}}

\newcommand{\alglab}[1]{\label{alg:#1}}
\renewcommand{\algref}[1]{\namedref{Algorithm}{alg:#1}}
\newcommand{\tablelab}[1]{\label{tab:#1}}

\newcommand{\deflab}[1]{\label{def:#1}}

\newcommand{\PPr}[1]{\ensuremath{\mathbf{Pr}\left[#1\right]}}

\newcommand{\Ex}[1]{\ensuremath{\mathbb{E}\left[#1\right]}}
\newcommand{\Var}[1]{\ensuremath{\mathbb{V}\left[#1\right]}}

\newcommand{\cO}[1]{\ensuremath{\mathcal{O}\left(#1\right)}}
\newcommand{\tO}[1]{\ensuremath{\tilde{\mathcal{O}}\left(#1\right)}}
\newcommand{\eps}{\varepsilon}

\def \SampleAndHold    {\mdef{\textsc{SampleAndHold}}}
\def \FullSampleAndHold    {\mdef{\textsc{FullSampleAndHold}}}

\def \calA    {\mdef{\mathcal{A}}}

\def \calD    {\mdef{\mathcal{D}}}
\def \calE    {\mdef{\mathcal{E}}}

\def \calO    {\mdef{\mathcal{O}}}

\def \calS    {\mdef{\mathcal{S}}}
\def \calT    {\mdef{\mathcal{T}}}

\newcommand{\mdef}[1]{{\ensuremath{#1}}\xspace}  

\DeclareMathOperator*{\polylog}{polylog}
\DeclareMathOperator*{\poly}{poly}

\DeclareMathOperator*{\median}{median}
\DeclareMathOperator*{\Uni}{Uni}

\newcommand{\flr}[1]{\mdef{\left\lfloor#1\right\rfloor}}              

\newcommand{\ignore}[1]{}

\newif\ifnotes\notestrue 
\ifnotes
\newcommand{\samson}[1]{\textcolor{purple}{{\bf (Samson:} {#1}{\bf ) }} \marginpar{\tiny\bf
             \begin{minipage}[t]{0.5in}
               \raggedright S:
            \end{minipage}}}            							
\else
\newcommand{\samson}[1]{}
\fi

\makeatletter
\renewcommand*{\@fnsymbol}[1]{\textcolor{mahogany}{\ensuremath{\ifcase#1\or *\or \dagger\or \ddagger\or
 \mathsection\or \triangledown\or \mathparagraph\or \|\or **\or \dagger\dagger
   \or \ddagger\ddagger \else\@ctrerr\fi}}}
\makeatother

\providecommand{\email}[1]{\href{mailto:#1}{\nolinkurl{#1}\xspace}}

\definecolor{mahogany}{rgb}{0.75, 0.25, 0.0}
\definecolor{darkblue}{rgb}{0.0, 0.0, 0.55}
\definecolor{darkpastelgreen}{rgb}{0.01, 0.75, 0.24}
\definecolor{bleudefrance}{rgb}{0.19, 0.55, 0.91}
\definecolor{darkgreen}{rgb}{0.0, 0.2, 0.13}
\definecolor{darkgoldenrod}{rgb}{0.72, 0.53, 0.04}
\definecolor{darkred}{rgb}{0.55, 0.0, 0.0}
\hypersetup{
     colorlinks   = true,
     citecolor    = mahogany,
		 linkcolor		= olive
}

\AtBeginDocument{%
  \DeclareFontShape{T1}{lmr}{m}{scit}{<->ssub*lmr/m/scsl}{}%
}